\documentclass[12pt]{amsart}
\usepackage{graphicx} 
\usepackage{amsmath, amscd, amsfonts, amssymb, color, bbm, bm,
braket, xcolor, setspace, cancel}
\usepackage{latexsym, graphicx, pstricks, rotating, enumerate}
\usepackage[pdftex,bookmarks,colorlinks]{hyperref}
\definecolor{darkblue}{rgb}{0.0,0.0,0.3}
\hypersetup{colorlinks=true,citecolor=darkblue,
unicode=true, pdftitle=upb-connection, pdfauthor=Ritabrata
Sengupta, urlcolor=darkblue} 
\usepackage[capitalise]{cleveref}
\usepackage{blindtext}
\theoremstyle{plain}
\numberwithin{equation}{section} 
\newtheorem{theorem}{Theorem}[section]   

\newtheorem{defin}[theorem]{Definition} 
\newtheorem{rem}[theorem]{Remark}
\newtheorem{eg}[theorem]{Example}
\newtheorem{lemma}[theorem]{Lemma}
\newtheorem{cor}[theorem]{Corollary}
\newtheorem{innercustomthm}{Theorem}
\newenvironment{customthm}[1] 
  {\renewcommand\theinnercustomthm{#1}\innercustomthm}
  {\endinnercustomthm}
\newcommand{\tr}{\mathrm{Tr\,}}

\newcommand{\ketbra}[1]{\ensuremath{\left|#1\right\rangle
\hspace{-3pt}\left\langle #1\right|}}

\begin{document}
\title[Useful variants and perturbations of CES and spans of UPB]{Useful variants and perturbations of completely entangled
subspaces and spans of unextendible product bases}
\author{Ritabrata Sengupta}
\address{Department of Mathematical Sciences, Indian Institute of
Science Education \& Research (IISER) Berhampur, Transit campus,
Berhampur 760 010, Ganjam, Odisha, India}
\email[Ritabrata Sengupta]{\href{mailto:rb@iiserbpr.ac.in}{rb@iiserbpr.ac.in}}

\author{Ajit Iqbal Singh}
\address{INSA Emeritus Scientist, Indian National Science Academy, Bahadur Shah Zafar Marg, New Delhi, 110 002, India}
\email[Ajit Iqbal Singh]{\href{mailto:ajitis@gmail.com}{ajitis@gmail.com}}

\begin{abstract}
Finite dimensional entanglement for pure states has been used extensively in quantum information theory. Depending  on the tensor product structure, even set of separable states can show non-intuitive characters. Two situations  are well studied in the literature, namely the unextendible product basis by Bennett et al [Phys. Rev. Lett. 82, 5385, (1999)], and completely entangled subspaces explicitly given by Parthasarathy in  [Proc. Indian Acad. Sci. Math. Sci. 114, 4 (2004)]. More recently Boyer, Liss, and Mor [Phys. Rev. A 95, 032308
(2017)]; Boyer and Mor [Preprints 2023080529, (2023)]; and Liss, Mor, and Winter [Lett. Math. Phys, 114, 86 (2024)] have  studied  spaces which have only finitely many pure product states. We carry this further and consider the problem of perturbing different spaces, such as the orthogonal complement of an unextendible product basis and also Parthasarathy's completely entangled spaces,  by taking linear spans with  specified product vectors. To this end, we develop methods and theory of variations and perturbations of the linear spans of certain unextendible product bases, their orthogonal complements, and also Parthasarathy's completely entangled sub-spaces.  Finally we give examples of perturbations with infinitely many pure product states. 

\smallskip
\vspace{3mm}

\noindent{\bf MSC Classification:} 81P15, 46C05, 46N50, 15A69, 
15A03, 81P16.\\
\vspace{2mm}
\newline{\bf Keywords:} Spans of unextendible product bases; quasi-completely entangled subspaces
(QCES); Parthasarathy space; perturbations of completely entangled subspaces by product vectors as QCES; infinitely many product vectors in double perturbations.
\vspace{5mm}
\newline In fond memory of our great teacher and source of inspiration Prof. K. R. Parthasarathy.
\end{abstract}

\maketitle
\section{Introduction, Basic Notation and Terminology}\label{s1}
Two well-known methods to construct spaces with no non-zero product vector, the so-called completely entangled subspaces $\mathcal{S}$ in finite-level quantum systems, are through orthogonal complements of  unextendible product bases, in short, UPB \cite{B2} and ideas from van der Monde determinants \cite{krp1,bhat,PhysRevA.87.064302,PhysRevA.90.062323}. We present similarities and contrasting situations for them by studying product vectors in the span $\mathcal{U}$ of UPB, and the span $\mathcal{F}$ of van der Monde vectors, and perturbing the corresponding orthogonal complement $\mathcal{S}$ by taking its span with product vectors in $\mathcal{U}$ and  $\mathcal{F}$ respectively. 

\par We first start with some basic notation and terminology. The state space of a quantum system is represented by a complex  separable Hilbert space $\mathcal{H}$. In this paper,  $\mathcal{H}$ is a finite dimensional complex Hilbert space, of dimension $d$ for some $d  \in \mathbb{N}$ usually with $d \ge 2$.  A unit vector $\ket{\psi} $ is a non-zero normalised vector in $\mathcal{H}$.  A \emph{state} acting on  $\mathcal{H} $ is a positive semidefinite operator $\rho$ with $\tr \rho=1$, $\tr X$ represents the trace of the operator $X$ on $\mathcal{H}$. A rank 1 state $\rho$ can be expressed as $\rho = \ketbra{\psi}$ and is called a \emph{pure state}. Hence without loss of generality the unit  vector $\ket{\psi}$ is also called pure state.  For the sake of convenience  we will treat the one dimensional space spanned by $\ket{\psi}$ as a nonzero vector. This understanding will be useful while counting one dimensional sub-spaces generated by (nonzero) product vectors, and referring to such spaces as \emph{product vectors  only}. 
\begin{defin}
A state $\rho$ acting on the tensor product Hilbert space $\mathcal{H} = \mathcal{H}_1 \otimes \mathcal{H}_2$, with dimensions $d_1 = \dim\mathcal{H}_1$ and  $ d_2  = \dim \mathcal{H}_2$, both $\ge 2$,  is said to be a \emph{separable state}, if it can be written as 
\[\rho= \sum_{s=1}^r p_s \sigma_s^{(1)} \otimes \sigma_s^{(2)},\qquad p_s \ge 0 ~\forall s, \quad   \sum_{s=1}^r p_s =1,\]
and $\sigma_s^{(j)}$ are states acting on $\mathcal{H}_j$ for $j =1,\,2$ and for every $s$. 
\par A state $\rho$ which can not be written in the above from is an \emph{entangled state}. This definition can easily be extended in multi-partite settings as well.   
\end{defin}
\par It is easy to check entanglement of a pure bipartite state, as tracing out one of the subsystems will give a mixed state. Moreover, if $\dim(\mathcal{H}_1 \otimes \mathcal{H}_2) =4$ or $6$, the partial transpose criteria is both necessary and sufficient to check separability of any given state. However in higher dimensions such criteria  are either only necessary, or only sufficient, or in case if both necessary and sufficient - they  are not easily computable. In fact, given an arbitrary state in an arbitrary bipartite quantum system the problem of checking whether the state is separable or entangled is in general a {\sf NP-hard} problem \cite{Gurvits:2003:CDC:780542.780545,MR2087945,MR2655512}. Multiple different methods using various mathematical techniques have been discussed in the paper \cite{MR3480684}. One such condition used for checking entanglement is called \emph{range criteria}, which is as follows. 
\begin{customthm}{A}\cite{H1} \label{cth:ra}
Let $\rho$ be a separable state in $\mathcal{H}$. Then the range of $\rho$ is spanned by a
set of product vectors.
\end{customthm}

\subsection{Unextendible product basis (UPB)} \label{ss1.1}
Let $k \ge 2, \, 1 \le j \le k$. Consider Hilbert spaces $\mathcal{H}_j = \mathbb{C}^{d_j}$ with $d_j \ge 2$.  Let  $\mathcal{H}$ be the tensor product $\bigotimes_{j=1}^k \mathcal{H}_j$. Then  $\dim \mathcal{H} =D= d_1\cdots d_k$. Let $N = d_1 + \cdots +d_k -k +1$.  We follow Bennett et al \cite{B2} for an introduction of UPB. 
\begin{defin} \cite{B2} \label{upb1}
An incomplete set of product vectors $\mathcal{B}$ in
the Hilbert space $\mathcal{H}=\bigotimes_{j=1}^k \mathcal{H}_j$ is called unextendible if the
space $\langle \mathcal{B}\rangle^\perp$ does not contain any non-zero product vector. The vectors in the set $\mathcal{B}$ are usually taken as orthonormal and $\mathcal{B}$ is called a UPB.
\end{defin}
\begin{rem} As we will show in \cref{th:2.1} in the next section, the word unextendible in the context of not necessarily mutually orthogonal set $\mathcal{B}$ is not really right because $\mathcal{C}=\mathcal{B} \cup \{\ket{\chi}\}$ in \cref{th:2.1}  extends $\mathcal{B}$ and possesses  the same property, viz., $\mathcal{C}^\perp$ has no non-zero product vector. 
This is confirmed again in the context of any $N$ van der Monde vectors as in \cref{sec3}.
\end{rem}

\par UPB's were constructed to give examples of a set of orthonormal product vectors, which can not be distinguished by LOCC \cite{B1}. It turns out that the normalised projection operator on the complementary subspace of the space spanned by  a UPB forms a bound entangled state. For details one can visit the original paper \cite{B2} where, apart from various examples, the  following important theorem was given. 
\begin{customthm}{B}\cite{B2} \label{cth:a}
Given a (mutually orthonormal) unextendible product basis: $\{\ket{\psi_s}: s = 1, \ldots, d\}$in the Hilbert space $\mathcal{H}=\bigotimes_{j=1}^k \mathcal{H}_j$ of dimension $D=d_1d_2 \cdots d_k$ as above, the state 
\[\rho = \frac{1}{D-d}\left[ I_D - \sum_{s=1}^d \ketbra{\psi_s}\right]\]
(where $I_D$ is the identity operator on $\mathcal{H}$) is an entangled state that is PPT.
\end{customthm}
\noindent The proof heavily depends on the orthogonality of the $\ket{\psi_j}$'s, and the range criteria is used for detecting entangled states \cite{H2}. 

\par In a follow-up article Di Vincenzo et al \cite{T2} constructed further examples of UPB's and generalisations of their earlier examples from  \cite{B2} for higher dimensions and/or multipartite settings. In this they also noticed that given $\mathcal{H}$ as above, there is a connection between the minimal cardinality of UPB $\mathcal{B}$ and the Ramsey number of graph theory. 
Alon and Lov{\'a}sz \cite{MR1840483} showed that  the minimum possible cardinality of such a set  $\mathcal{B}$ is precisely $N =1+  \sum_{j=1}^k (d_j -1)$  for every sequence of integers $d_1, d_2, \ldots, d_k \ge 2$  unless either (i) $k=2$ and $2 \in \{d_1, d_2\}$; or (ii) $N$ is odd and at least one of the $d_j$ is even. In each of these two cases, the minimum cardinality of the corresponding $\mathcal{B}$'s is strictly bigger than $N$. 

\par It should be noted that Theorem \ref{cth:a} gives perhaps so far the only systematic way to generate bound entangled states (which are positive under partial transpose, i.e.,  PPT).  In fact, this also gave  a method to construct indecomposable positive maps which are not completely positive \cite{T1}.  Further work gives that in two qutrit system, i.e. when $k=2$ and $\dim \mathcal{H}_1= \dim \mathcal{H}_2=3$, all the rank 4 PPT entangled states are coming from UPB systems \cite{MR2907637} described in \cite{B2}. Johnston \cite{PhysRevA.87.064302} constructed a subspace $\mathcal{S}$  as large as possible in terms of its dimension, of a bipartite system ($\mathcal{H}=\mathcal{H}_1 \otimes \mathcal{H}_2$ with $\dim \mathcal{H}_1 = d_1$ and $\dim\mathcal{H}_2 =d_2$) such that every state with support in $\mathcal{S}$ is not-positive under partial transpose (in short,  NPT). We emphasise that the UPB's  can not be distinguished by LOCC transformations. This has been used in various quantum tasks as in \cite{PhysRevLett.122.040403,MR3390497}. For latest developments and new constructions, with weaker constrains such as not necessarily orthogonal etc. one can see \cite{saro_rb1,saro_rb2,MR4329848}. 
\subsection{Completely entangled subspace}\label{ss1.2}
\par During this time, an almost complementary concept was proposed first by Wallach \cite{MR1947343} and later by Parthasarathy \cite{krp1} with concrete construction, viz., that of completely entangled subspace, as defined below. 
\begin{defin}\label{mes1}
A non-trivial subspace $\mathcal{S} \subset \mathcal{H}=\mathcal{H}_1 \otimes \cdots \otimes \mathcal{H}_k$ is said to be  completely entangled subspace,  if $\mathcal{S}$ does not contain any nonzero product  vector of the form $\ket{u_1} \otimes \cdots \otimes\ket{u_k}$ where $\ket{u_j}\in \mathcal{H}_j, \, j=1,\ldots, k$. 
\newline At times if no confusion can arise, it will be abbreviated  as \emph{CES}. 
\end{defin}
Let $\mathcal{E}$ be the set of all completely entangled subspaces of $\mathcal{H}$. The set is non-empty simply because the one-dimensional subspaces spanned by an entangled vector, are present in $\mathcal{E}$. Indeed there is a space $\mathcal{S} \in \mathcal{E}$ of maximal dimension $D-N$ as stated below. 
\begin{customthm}{C}\cite{krp1} \label{cth:b}
There is a space $\mathcal{S}_{max} \in \mathcal{E}$ such that 
\begin{equation}\label{eq1} 
\dim \mathcal{S}_{max} =  \max_{\mathcal{S}\in \mathcal{E}} \dim \mathcal{S} =d_1 \cdots d_k -(d_1 + \cdots+ d_k) +k -1.  
\end{equation}
\end{customthm}
One side of the proof follows from the explicit construction of $\mathcal{S}_{max}$ whose dimension is  at least the right-hand side of the \cref{eq1}. The other direction follows from the fact that  a subspace of dimension $\ge D-N+1 $ contains a non-zero product vector. 
\par The paper \cite{krp1} also contains an explicit construction of a basis for bipartite $\mathbb{C}^d \otimes \mathbb{C}^d$ space. The construction was further modified by Bhat \cite{bhat} for a multipartite system. Johnston \cite{PhysRevA.87.064302} later concentrated on constructing a completely entangled subspace of $\mathbb{C}^{d_1} \otimes \mathbb{C}^{d_2}$ of dimension $(d_1-1)(d_2-1)$ for bipartite systems such that every density matrix with range contained in it is NPT. Later in \cite{PhysRevA.90.062323} the present authors along with Arvind,  showed that this construction in  \cite{PhysRevA.87.064302} can also be derived from that of  Parthasarathy \cite{krp1} and Bhat \cite{bhat}. Furthermore, it was shown \cite{PhysRevA.90.062323} that the projection on a completely entangled subspace $\mathcal{S}$ of maximum dimension obtained in \cite{krp1}  in a multipartite quantum system is also not positive under partial transpose, a large number of positive operators with  range in $\mathcal{S}$ also have the same property.
\par Behaviour of the linear span, in short, the span of UPB's considered, their orthogonal complements,  and  $\mathcal{S}_{\max}$  can be quite different. In this paper, the space  $\mathcal{S}_{\max}$ will be called \emph{Parthasarathy space} and will be denoted by $\mathcal{S_P}$ to honour K. R. Parthasarathy for his contributions to the topic. The primary purpose of this paper is to display that along with   a better understanding of the setup and further generalisations. 
\subsection{Quasi completely entangled subspace}\label{ss1.3}

\par  In case of a UPB, the complementary subspace does not contain any non-zero product vector, and hence forms a CES.  In Ref \cite{krp1} the maximal possible dimension  CES has been worked in detail, as mentioned earlier in  \cref{ss1.2}.  In a recent series of papers like Boyer, Liss, and Mor \cite{blm17}, Boyer, and Mor \cite{mor1} and Liss, Mor, and Winter \cite{mor2}, a similar problem is considered with slight change. Effectively, the authors are considering the spaces which are essentially CES but containing finitely many pure product states. We call these subspaces $T$ \emph{Quasi completely entangled subspace},  in short, QCES. We let $t$ be the dimension of $T$ and, following Ref. \cite{mor2}, we call  the number of pure product states in $T$  as its \emph{product index}, and denote it by $\tau$. We shall treat the elements in the one-dimensional subspace generated by a non-zero vector $\ket{\zeta}$ as the vector $\ket{\zeta}$ only, particularly so while counting the number of product vectors in such subspaces. Liss et al  \cite{mor2} asks the question of relationship between $t$ and $\tau$, in particular $\max \tau$ for a given $t$. 

\par Any state $\rho$ with support in a QCES with specified product vectors, say $\mathcal{P}$  has the form $\sum_{\xi} p_\xi \ketbra{\xi}$ with $\ket{\xi}$ varying over $\mathcal{P}$, all $p_\xi \ge0$ with $\sum p_\xi =1$. One can see such instances, for example, see Remark 2.9  of Ref \cite{MR3480684}. Members of $\mathcal{P}$ may or may not be  mutually orthogonal.
 
\subsection{Polynomial representation of  multipartite systems and quantum entanglement} \label{ss1.4}
\par We use the well-known (see, for instance, Ref \cite{MR1947343,MR3635750}) polynomial representation of  multipartite syatems to find more product vectors  in  certain spaces. For a Hilbert space $\mathcal{H}_A$ we represent the basis vectors $\ket{0},\, \ket{1},\, \ket{2}$ etc as $1,\, X, \, X^2$  etc. Similarly for the space $\mathcal {H}_B$ we may represent $ \ket{0} \mapsto 1, \, \ket{1} \mapsto Y, \,\ket{2} \mapsto Y^2 $ etc., for $\mathcal{H}_C$ in terms of $Z$,  and etc. Thus a vector in the space $\mathcal{H}_A$ is represented as a polynomial in $X$ and a vector in   $\mathcal{H}_B$ is represented as a polynomial in $Y$. 
This representation was well used to study multipartite quantum entanglement in  \cite{MR3635750} and later utilised  in \cite{singh23}. 

\section{Examples of quasi-completely entangled spaces related to UPBs}\label{sec2}

\par In this section we discuss in detail spaces related to TILES unextendible product basis and the one in three qubits given in Ref \cite{B2} with relation to variants and perturbations of their spans and the completely entangled spaces given by their orthogonal complements. This will enable us to express the space as a direct sum of two QCES as noted as a corollary at its first occurrence in \cref{cor.2.3}.
\subsection{TILES UPB related sub-spces}\label{ssec2.1}
\par We first start with the UPB introduced in Ref \cite{B2}, known as TILES in $\mathbb{C}^3 \otimes \mathbb{C}^3$. For completeness, let us define the vectors as follows using the short form $\ket{\xi} \ket{\eta}$ for $\ket{\xi} \otimes \ket{\eta}$.
\begin{eqnarray} \label{upbe1}
\ket{\psi_0} &=& \frac{1}{\sqrt{2}} \ket{0} (\ket{0} - \ket{1}), \nonumber\\
\ket{\psi_1} &=& \frac{1}{\sqrt{2}} \ket{2} (\ket{1} - \ket{2}), \nonumber\\
\ket{\psi_2} &=& \frac{1}{\sqrt{2}} (\ket{0} - \ket{1}) \ket{2},  \\
\ket{\psi_3} &=& \frac{1}{\sqrt{2}} (\ket{2} - \ket{1}) \ket{0}, \nonumber\\
\ket{\psi_4} &=& \frac{1}{3} (\ket{0} + \ket{1} + \ket{2}) (\ket{0} + \ket{1} + \ket{2}). \nonumber
\end{eqnarray}
\begin{theorem}\label{th:2.1}
The span $\mathcal{U}$ of TILES UPB is a 5-dimensional quasi completely entangled subspace (QCES) of product index 6.  The sixth product vector is the vector 
\[\ket{\chi} = \frac{\sqrt{2}}{3} (\ket{\psi_0} +  \ket{\psi_1} + \ket{\psi_2} + \ket{\psi_3} )+ \frac{1}{3} \ket{\psi_4}. \]
\end{theorem}
\begin{proof}
Consider $\bm{\alpha} =( \alpha_0,\alpha_1,\alpha_2,\alpha_3,\alpha_4) $ with $\#\{j: \alpha_j \neq 0\} \ge 2$. Since $\ket{\psi_j}$'s are mutually orthogonal, 
\[\ket{\xi} = \alpha_0\sqrt{2} \ket{\psi_0} + \alpha_1\sqrt{2} \ket{\psi_1} + \alpha_2\sqrt{2} \ket{\psi_2} + \alpha_3\sqrt{2} \ket{\psi_3} + 3\alpha_4\ket{\psi_4} \neq 0.\]
Rearranging the terms, let us write 
\[\ket{\xi} =  \ket{0} \ket{u} + \ket{1} \ket{v} + \ket{2} \ket{w},\]
where 
\begin{eqnarray}\label{upbc1}
\ket{u} &=& (\alpha_0 + \alpha_4) \ket{0} + (-\alpha_0 + \alpha_4) \ket{1} + (\alpha_2 + \alpha_4) \ket{2},\nonumber\\
\ket{v} &=& (-\alpha_3 + \alpha_4) \ket{0} +  \alpha_4 \ket{1} + (-\alpha_2 + \alpha_4) \ket{2},\\
\ket{w} &=& (\alpha_3 + \alpha_4) \ket{0} + (-\alpha_1 + \alpha_4) \ket{1} + (\alpha_1 + \alpha_4) \ket{2} \nonumber.
\end{eqnarray}
Since $\ket{\xi} \neq 0$, at least one of $\ket{u}, \, \ket{v}$, and $\ket{w}$ must be nonzero. Now $\ket{\xi} $ is a product vector, if and only if $\ket{u} = \alpha \ket{z}$, $\ket{v} = \beta \ket{z}$, and $\ket{w} = \gamma \ket{z}$ for some nonzero vector $\ket{z}$, with $\alpha, \, \beta, \, \gamma \in \mathbb{C}$ not all three of them equal to $0$. 

\par {\bf Case $\alpha_4 = 0$}. Then at least two of $ \alpha_0,\, \alpha_1,\, \alpha_2,\, \alpha_3$ are non-zero. Also the above equation \eqref{upbc1} takes the form 
\begin{eqnarray*}
\ket{u} &=& \alpha_0  \ket{0} -\alpha_0  \ket{1} + \alpha_2  \ket{2},\\
\ket{v} &=& -\alpha_3 \ket{0} -\alpha_2  \ket{2},\\
\ket{w} &=& \alpha_3  \ket{0}  -\alpha_1 \ket{1} + \alpha_1  \ket{2} .
\end{eqnarray*}
Now if $\ket{v}=0$ i.e. $\alpha_2 = 0 = \alpha_3$, then $\alpha_0, \, \alpha_1 \neq 0$ and $\ket{\xi}$  is not a product vector. \\
If $\ket{v} \neq 0$ i.e. $\alpha_2, \, \alpha_3$ not both of them are zero, we may take $\ket{z} = \ket{v}$ and $\beta =1$. Then $\ket{\xi}$ is a product vector if and only if
\begin{eqnarray*}
(\alpha_0, -\alpha_0, \alpha_2) &=& \alpha (-\alpha_3, 0, -\alpha_2),\\
(\alpha_3, -\alpha_1, \alpha_1) &=& \gamma (-\alpha_3, 0, -\alpha_2).
\end{eqnarray*}
A simple calculation shows that the above is true if and only if  $\alpha_0=0=\alpha_1$ and  $ \alpha_2$ or $ \alpha_3$ is equal to zero according as $\alpha$ is zero or non-zero. This is contrary to our assumption that at least two of the $\alpha_j$'s should be non-zero. \\

\noindent {\bf Case} $\alpha_4 \neq 0$. Let $\alpha_j = \beta_j \alpha_4$ for $j=0,1,2,3$. Then at least one of the $\beta_j$'s is non-zero. Moreover 
\begin{eqnarray*}
&& \ket{u'} = \frac{1}{\alpha_4} \ket{u} = (1+ \beta_0) \ket{0} + (1- \beta_0) \ket{1}  + (1+ \beta_2) \ket{2}, \\
&& \ket{v'} = \frac{1}{\alpha_4} \ket{v} = (1- \beta_3) \ket{0} +  \ket{1}  + (1- \beta_2) \ket{2}, \\
&& \ket{w'} = \frac{1}{\alpha_4} \ket{w} = (1+ \beta_3) \ket{0} + (1- \beta_1) \ket{1}  + (1+ \beta_1) \ket{2}.
\end{eqnarray*} 
For $\delta \in \mathbb{C}$ both $1 + \delta$ and $1 - \delta$ can not be zero. So $\ket{u'} \neq 0 \neq \ket{v'}, \, \ket{w'} \neq 0$. So $\ket{\xi}$ is a product vector if and only if 
\begin{eqnarray*}
\ket{u'} &=& (1- \beta_0) \ket{v'} \\
\ket{w'} &=& (1- \beta_1) \ket{v'},
\end{eqnarray*}
if and only if $\beta_0 \neq 1 \neq \beta_1$, together with
\begin{eqnarray}
1 + \beta_0 &=& (1 - \beta_0) (1 - \beta_3) \label{up-a}\\
1 + \beta_2 &=& (1 - \beta_0) (1 - \beta_2) \label{up-b}\\
1 + \beta_3 &=& (1 - \beta_1) (1 - \beta_3) \label{up-c}\\
1 + \beta_1 &=& (1 - \beta_1) (1 - \beta_2). \label{up-d}
\end{eqnarray}
Let us consider the system of equations \eqref{up-a} --\eqref{up-d}. Then $\beta_j \neq 1$ for $ j =0,1,2,3$.  \eqref{up-a} gives 
\[1 + \beta_3  = (\beta_3 -1) +2 = \frac{1 -3 \beta_0}{1 - \beta_0}.\]
Substituting in \eqref{up-c}  we get 
\[\frac{1 -3 \beta_0}{1 - \beta_0} = (1 - \beta_1) \frac{1+\beta_0}{1 - \beta_0}.\]
which after simplification gives 
\begin{equation} \label{up-e}
\beta_0 \beta_1 + \beta_1 - 4 \beta_0 =0.
\end{equation}
In a similar way \eqref{up-d} and \eqref{up-b} give 
\begin{equation} \label{up-f}
\beta_0 \beta_1 + \beta_0 - 4 \beta_1 =0.
\end{equation}
\eqref{up-e} and \eqref{up-f} give $\beta_1 = \beta_0$. Substituting in \eqref{up-e} we get $\beta_0^2 - 3\beta_0 =0$ which has solutions $\beta_0 = 3$ or $0$. If $\beta_0 = \beta_1= 0$ then combined with \eqref{up-a} and \eqref{up-d} we get $\beta_3 =0 =\beta_2$. But at least one of the $\beta_j$'s should be nonzero.   So we must have $\beta_0 = \beta_1= 3$. \eqref{up-a} and \eqref{up-b} then gives $\beta_3 = 3 =\beta_2$. On the other hand, $\beta_0=\beta_1=\beta_2 = \beta_3= 3$ satisfies the system. In this case 
\[\ket{\xi} = \alpha_4( (2 \ket{0} - \ket{1} + 2 \ket{2} ) (2 \ket{0} - \ket{1} + 2 \ket{2}),\]
which becomes a unit vector
$\ket{\chi}$, say  by taking $\alpha_4=\frac{1}{9}$.
Note that $\ket{\chi}$  
is a unit vector in the linear span $\mathcal{U}$ of $\ket{\psi_0}, \, \ket{\psi_1}, \, \ket{\psi_2},\, \ket{\psi_3}, \, \ket{\psi_4}$, and $\mathcal{U}$ contains 6 one dimensional subspaces containing product vectors.  
\end{proof}
\begin{rem}
\begin{enumerate}[(i)]
\item Any five vectors from the set  $\mathcal{C} = \mathcal{B} \cup \{\ket{\chi}\}$ span $\mathcal{U}$.
\item For $1 \le p \le 4$, any $p$ vectors from  $\mathcal{C}$ span a $p$ dimensional QCES of  product index $p$.   
\end{enumerate}
\end{rem}
\subsection{Perturbation of $\mathcal{S_U}= \mathcal{U}^\perp$ by vectors from $\mathcal{C}$} \label{ssec2.2}

 \par We find polynomial representation handy and notationally simple to deal with proofs in this paper, particularly for low dimensions. 
  \begin{rem}[Perturbation related to TILES UPB] \label{rem2.2}
  Using the polynomial representation rules,  we  represent the vectors of TILES UPB written in \eqref{upbe1} without using the normalisation factors. 
\begin{enumerate}[(i)]
\item   The corresponding unnormalised vectors take the polynomial form 
 \begin{eqnarray*}\label{upbe2}
F_0(X,Y)  &=& 1(1 -Y), \nonumber\\
F_1(X,Y)  &=& X^2 (Y^2 -Y) = X^2 Y^2 - X^2 Y, \nonumber \\
F_2(X,Y)  &=& (1 -X) Y^2 = Y^2 - X Y^2, \\
F_3(X,Y)  &=& (X^2 -X) 1 = X^2 - X, \nonumber \\
F_4(X,Y)  &=& (1 +X + X^2) (1 +Y +Y^2) \nonumber \\
&=& 1 + X +Y + X^2 + XY + Y^2 + X^2 Y + X Y^2 +X^2 Y^2.\nonumber 
 \end{eqnarray*}
\item  
 Any point $\ket{\eta}$ of $\mathcal{U}$ has polynomial representation 
  \begin{eqnarray*}
 F(X,Y) &=& \alpha_0 (1 - Y) + \alpha_1 (X^2 Y^2 - X^2 Y) + \alpha_2 (Y^2 - X Y^2) \\
 && + \alpha_3 (X^2 - X) + \alpha_4 (1 +X + X^2) (1 +Y +Y^2).
 \end{eqnarray*}
Here the coefficients $\alpha_j$s are the scalars. 
\item Let $\mathcal{S}_\mathcal{U}=\mathcal{U}^\perp$.  So the polynomial representation $\sum_{s,t=0}^2 a_{st} X^s Y^t$, say $G(X,Y)$ of $\ket{\xi} \in \mathcal{S}_\mathcal{U}$  should satisfy 
 \begin{eqnarray*}
 a_{00} - a_{01} &=& 0 \quad \text{ i.e., } \quad a_{01} = a_{00}, \\
 a_{22} - a_{21} &=& 0 \quad \text{ i.e., } \quad a_{22} = a_{21}, \\
 a_{02} - a_{12} &=& 0 \quad \text{ i.e., } \quad a_{02} = a_{12},\\
 a_{20} - a_{10} &=& 0 \quad \text{ i.e., } \quad a_{20} = a_{10}, \\
 \sum_{s,t =0}^2 a_{st} &=&0 \quad \text{ i.e., } \quad a_{11} = -\sum_{\substack{0 \le s, t \le 2\\(s,t) \neq (1,1)} } a_{st}.
 \end{eqnarray*}
So, 
 \begin{eqnarray*}
 G(X,Y) &=& \gamma_0 (1 + Y) + \gamma_1 (X^2 Y^2 + X^2 Y) + \gamma_2 (Y^2 + X Y^2) \\
 && + \gamma_3 (X^2 + X) - 2 (\gamma_0 + \gamma_1 + \gamma_2 + \gamma_3) XY, \quad \text{say.}
 \end{eqnarray*}
 
\item  We would like to identify $\bm{\gamma} = (\gamma_j)_{j=0}^3$ and $\bm{\alpha} = (\alpha_j)_{j=0}^4$ such that 
 $\Sigma(X, Y)  = G(X,Y) + F(X,Y)$ is a non-zero product vector, , i.e., it has the form
\begin{equation}
\Sigma (X,Y)=p(X) q(Y)
\end{equation}
 for some polynomials $p$ and $q$.
 \par This is not possible if $\bm{\alpha}= \bm{0}$ because $\mathcal{S}_\mathcal{U}$ is completely entangled. Also for $\bm{\gamma}=\bm{0}$  we know exactly six product vectors. So, we consider the case for $\bm{\alpha} \neq  \bm{0}$ and $\bm{\gamma} \neq \bm{0}$. 
 \end{enumerate}
 \end{rem}
  \par We begin with the case $\bm{\alpha} = (1,0,0,0,0)$. In other words we consider the subspace  $\mathcal{S}_{\mathcal{U},0}$ spanned by $\mathcal{S}_\mathcal{U}$ and $\ket{\psi_0}$. So 
 \begin{eqnarray*}
 \Sigma(X,Y) &=& 1\left( (1+\gamma_0) +(-1 + \gamma_0) Y + \gamma_2 Y^2 \right) \\
 && +  X \left( \gamma_3 -2 ( \gamma_0 + \gamma_1 + \gamma_2 + \gamma_3) Y + \gamma_2 Y^2 \right) \\
&&  + X^2 \left ( \gamma_3 + \gamma_1 Y + \gamma_1 Y^2 \right)\\
 &=& 1 g_0(Y) + X g_1(Y) + X^2 g_2(Y), \text{ say.}
 \end{eqnarray*}
Now both $1+\gamma_0$ and $-1 + \gamma_0$ can not be zero together. So $g_0(Y) \neq 0$. Hence $\Sigma(X,Y)$ gives a product vector $\ket{\zeta}$  if and only if 
 \begin{eqnarray}\label{upbe3b}
 g_1(Y) &=& \nu_1 g_0(Y), \nonumber\\
  g_2(Y) &=& \nu_2 g_0(Y)
  \end{eqnarray} 
  for some $(\nu_1, \nu_2) \in \mathbb{C}^2$ and in this case 
\[\Sigma (X,Y) = (1 + \nu_1 X + \nu_2 X^2) g_0(Y).\]
\begin{theorem}
Perturbation   $\mathcal{S}_{\mathcal{U}, 0}$  of $\mathcal{S}_{\mathcal{U}}$  by taking its linear span with $\ket{\psi_0}$ is a QCES with product index $6$ with specified product vectors given in the proof.
\end{theorem}
\begin{proof}
We continue from \eqref{upbe3b} above, which holds if and only if 
\begin{alignat}{3}\label{upbe3c}
   \gamma_3 &= \nu_1 (1 + \gamma_0),      \qquad      & -2(\gamma_0 + \gamma_1 + \gamma_2 + \gamma_3) &= \nu_1(-1 + \gamma_0), & \qquad\gamma_2 &= \nu_1 \gamma_2; \\
   \gamma_3 &= \nu_2 (1 + \gamma_0),  & \gamma_1 &= \nu_2(-1+\gamma_0), & \gamma_1 &= \nu_2 \gamma_2.
    \notag
\end{alignat}
Now $\nu_1=0$ forces $\gamma_3=0, \, \gamma_2=0, \, \gamma_1=0,\, \gamma_0=0,$ which is not possible. Therefore $\nu_1 \neq 0$.   Now $\nu_2=0$ forces $\gamma_3=0, \, \gamma_1=0$. Because $\nu_1 \neq 0$,  we get $\gamma_0 = -1$. So $-2(-1+\gamma_2) =-2 \nu_1$, i.e. $\nu_1 =-1 + \gamma_2 \neq 0$.  So $\gamma_2 = -\gamma_2 + \gamma_2^2$, i.e. $\gamma_2 =0$ or $2$. Hence $\nu_1 = -1$ or $1$ according to    $\gamma_2 =0$ or $2$. So we have two solutions, 
\begin{alignat*}{2}
   (\nu_1, \nu_2) &= (-1,0),  \qquad       & \bm{ \gamma} &= (-1,0,0,0), \text{ and } \\
   (\nu_1, \nu_2) &= (1,0) ,        & \bm{ \gamma} &= (-1,0,2,0).
\end{alignat*}
Accordingly we get 
\begin{eqnarray}\label{2.13}
   \Sigma(X,Y) &=& (1 -X) (-2Y), \qquad \text{ or }\\
   \Sigma(X,Y) &=& (1 + X) (-2Y + 2Y^2).  \notag
\end{eqnarray}
Now suppose $\nu_2 \neq 0 \neq \nu_1$. From \eqref{upbe3c} we get $(\gamma_1, \gamma_3) \neq (0,0),\, \gamma_1 =0$ if and only if $\gamma_2=0$.
\par  {\bf For the case}  $\gamma_3 =0$  \eqref{upbe3c} gives $\gamma_0=-1, \, \gamma_1 = -2\nu_2 \neq 0$. So $\gamma_2 =-2$ and $\nu_1=1$. Putting the values in the equations \eqref{upbe3c} we get $\nu_2=-2,\, \gamma_0 =-1, \, \gamma_1 =4, \, \gamma_2 =-2,\, \gamma_3 =0$. Thus the polynomial representation takes the form 
\begin{equation}\label{2.14}
\Sigma(X,Y)=  (1 + X - 2X^2) (-2Y - 2Y^2) = -2 (1 + X - 2X^2) (Y + Y^2).
\end{equation}

\par {\bf Case} $\gamma_3 \neq 0$. Then by \eqref{upbe3c} we get $\nu_1 = \nu_2 \neq 0,\, 1 + \gamma_0 \neq 0$.  Two possible subcases may arise. 
\newline If $\gamma_2 =0$ then $\gamma_1 =0,\, \gamma_0 =1$. So by \eqref{upbe3c} $\gamma_3=-1,\, \nu_1 = -\frac{1}{2}$. In this subcase 
\begin{equation}\label{2.15}
\Sigma(X,Y)=  \left(1  -\frac{1}{2} X -\frac{1}{2}X^2\right) 2 = 2 - X - X^2.
\end{equation}
\newline If $\gamma_2 \neq 0$ then \eqref{upbe3c} gives $\nu_1 =1,\, \gamma_3 = 1+ \gamma_0, \, \gamma_1 = \gamma_2 = -1 +\gamma_0$, and $\gamma_1 =-2(\gamma_0 + \gamma_1 + \gamma_2 + \gamma_3)$. Further calculations give $\bm{\gamma} = \left( \frac{1}{3}, -\frac{2}{3}, -\frac{2}{3}, \frac{4}{3} \right)$. Corresponding polynomial form is 
\begin{equation}\label{2.16}
\Sigma(X,Y)=  \left(1  + X +X^2\right) \left( \frac{4}{3} - \frac{2}{3} Y - \frac{2}{3} Y^2 \right).
\end{equation}
Hence we have exactly six (normalised) product vectors in $\mathcal{S}_{\mathcal{U},0}$ as follows,
\begin{alignat*}{2}
\ket{\psi_0}=  \frac{1}{\sqrt{2}} \ket{0} \otimes(\ket{0} - \ket{1}), \quad & \frac{1}{\sqrt{2}}(\ket{0} - \ket{1}) \otimes \ket{1},\\ \frac{1}{2}(\ket{0} + \ket{1}) \otimes (\ket{1} - \ket{2}),\quad &
\frac{1}{2\sqrt{3}}(\ket{0} + \ket{1} - 2 \ket{2}) \otimes (\ket{1} +  \ket{2}),\\
  \frac{1}{\sqrt{6}}(2\ket{0}- \ket{1} - \ket{2}) \otimes \ket{0}, \quad &
\frac{1}{3\sqrt{2}}(\ket{0} + \ket{1} +  \ket{2}) \otimes \left( 2 \ket{0} -  \ket{1} -  \ket{2} \right).
\end{alignat*}
\end{proof}
The orthogonal complement of this space $\mathcal{S}_{\mathcal{U},0}$ is the linear span of $\ket{\psi_1}, \,\ket{\psi_2}, \, \ket{\psi_3}, \, \ket{\psi_4}$ and has exactly these four product states. So we have the following Corollary. 
\begin{cor}\label{cor.2.3}
$\mathcal{H}$ is expressible as direct sum of two QCES, viz., the span of $\{ \ket{\psi_1}, \,\ket{\psi_2}, \, \ket{\psi_3}, \, \ket{\psi_4} \}$, and $\mathcal{S}_{\mathcal{U},0}$. 
\end{cor}
\begin{rem}
Similar results can be proved for other members of $\mathcal{C}= \mathcal{B} \cup \{\ket{\chi}\}$. 
\end{rem}
\par Computations for linear span of $\mathcal{S}_{\mathcal{U}}$ and $\ket{\psi_j}$, (for $j=1,\,2,\,3$) are similar. But for $\ket{\psi_4}$ it takes a different form.
\begin{theorem}\label{th2.4}
The span $\mathcal{S}_{\mathcal{U},4}$ of $\mathcal{S}_{\mathcal{U}}$ and $\ket{\psi_4}$ is QCES of product index 6 with new product vectors specified in the proof. 
\end{theorem} 
\begin{proof}
\par  Indeed, for $\bm{\gamma} =  (\gamma_0, \gamma_1, \gamma_2, \gamma_3) \in \mathbb{C}^4$, 
\begin{align*}
G(X,Y) =& \gamma_0 (1+Y) + \gamma_1 (X^2Y^2 +X^2 Y) +\gamma_2 (Y^2 + XY^2) \\
&+\gamma_3(X^2 +X) - 2(\gamma_0+\gamma_1+\gamma_2+\gamma_3) XY.
\end{align*}
is non-zero if and only if $\gamma_j \neq 0$ for some $j=0,\,1,\,2,\,3$, say for $j'$, and in that case either one more form $\gamma_j$'s is non-zero or else $- 2(\gamma_0+\gamma_1+\gamma_2+\gamma_3)$ is non-zero, and thus the number of terms in $G(X,Y)$ is greater than or equal to $3$. 
\par Consider $\bm{\gamma} =(1,1,1,1)$. Then 
\[G(X,Y) =F_4(X,Y)-9XY.\]
So $\mathcal{S}_{\mathcal{U}, 4}$ contains two different product vectors at least,  viz., $\ket{\psi_4}$ and $\ket{11}$.
\par A general element of $\mathcal{S}_{\mathcal{U},4}$ has the form 
\[\Sigma(X,Y) = \gamma_0 (1+Y) + \gamma_1 (X^2Y^2 +X^2 Y) +\gamma_2 (Y^2 + XY^2) 
+\gamma_3(X^2 +X) +\gamma_4 XY,\]
with $\bm{\tilde{\gamma}}=(\gamma_0, \gamma_1,\gamma_2,\gamma_3, \gamma_4) \in \mathbb{C}^5$. 
$\ket{\psi_4}$ and $\ket{11}$ are given by $(1,1,1,1,1)$ and $(0,0,0,0,1)$ respectively. We look for more product vectors for the case $\tilde{\bm{\gamma}}$ not a scalar multiple of $(0,0,0,0,1)$ or of  $(1,1,1,1,1)$.  Now
\begin{eqnarray*}
\Sigma(X,Y) &=& 1.(\gamma_0(1+Y) +\gamma_2 Y^2) +X(\gamma_3 + \gamma_4 Y +\gamma_2 Y^2) +X^2 (\gamma_3 +\gamma_1(Y+Y^2))\\
&=& 1. h_0(Y) + Xh_1(Y) +X^2 h_2(Y), \quad \text{ say.}
\end{eqnarray*}
It gives a non-zero product vector if and only if for some $h(Y) \neq 0,\, \bm{\beta}=(\beta_0,\beta_1,\beta_2) \neq(0,0,0),\, h_j(Y)=\beta_j h(Y),\, j=0,\,1,\,2$. Let this be the situation. 
\par Consider the {\bf case} $\gamma_0\neq 0$. Then $h_0(Y) \neq 0$ and we may take $h(Y) = h_0(Y),\, \beta_0=1$. This happens if and only if 
\begin{alignat*}{2}
\gamma_3 &= \gamma_4 =\beta_1\gamma_0, \quad & \gamma_2 &=\beta_1\gamma_2,\\
\gamma_3 &= \beta_2 \gamma_0, & \gamma_1 &= \beta_2 \gamma_0=\beta_2 \gamma_2.
\end{alignat*}
Now if $\beta_1 =0$, then $\gamma_2 =0 = \gamma_3 =\gamma_4$ and therefore $\gamma_1=0,\, \beta_2 =0$ which gives 
\[\Sigma(X,Y) = \gamma_0 (1+Y).\]

\par On the other hand, if $\beta_1 \neq 0$, then $\gamma_3 = \gamma_4 \neq 0$. 
So $\beta_1\neq 0 \neq \beta_2, \, \gamma_1 \neq 0$ and $\gamma_2 = \gamma_0 \neq 0$. Hence  $\beta_1 =1$ and therefore $\gamma_3 =\gamma_4 =\gamma_0 =\gamma_2$ and $\gamma_1 =\gamma_0$. This gives the case $\bm{\gamma}=(1,1,1,1,1)$ which is not so.  

\par   We now come to the {\bf case} $\gamma_0=0$.\\ 
\noindent {\bf Sub-case} $\gamma_2 \neq 0$. Then $h_0(Y) = \gamma_2 Y^2 \neq 0$ and we may take $h(Y)=h_0(Y),\, \beta_0=1$. Then the  conditions relating $h_1$ and $h_2$ hold if and only if $\gamma_3=0=\gamma_4=\gamma_1$ which gives $h_1(Y) =\gamma_2Y^2,\, h_2(Y)=0$. Thus $\bm{\gamma} =(0,0,\gamma_2,0,0)$ and 
\[\Sigma(X,Y)= \gamma_2(1+X)Y^2.\]
{\bf Sub-case} $\gamma_2=0$. Then $h_0(Y) =0$. So $(\gamma_1,\gamma_3,\gamma_4) \neq (0,0,0)$, this has three  sub-cases as fellows. \\
\noindent {\bf Sub-case A:} $\gamma_1 \neq 0$. Then $h_2(Y) \neq 0$ So we may take $h(Y) =h_2(Y),\, \beta_2 =1$. Then $\gamma_4=\gamma_2=\beta_1 \gamma_1,\, \gamma_3 =\beta_1\gamma_3$. So $\gamma_4=0,\, \beta_1=0$ and hence $\gamma_3=0$. The polynomial representation becomes 
\[\Sigma(X,Y) =\gamma_1 X^2 (Y+Y^2).\]
{\bf Sub-case B:} $\gamma_1=0 \neq \gamma_3$. So $h_2(Y)=\gamma_3 \neq 0$ So we may take $h(Y) =h_2(Y),\, \beta_2=1$. Then we have $\gamma_4 =\gamma_2 =\gamma_0 =0$. This gives $\bm{\tilde{\gamma}}=(0,0,0,\gamma_3,0)$. So 
\[\Sigma(X,Y) =\gamma_3(X+X^2).\]
{\bf Sub-case C:} $\gamma_1=0=\gamma_3,\,\gamma_4 \neq 0$. So $\bm{\tilde{\gamma}}=(0,0,0,0,\gamma_4)$. But that has been already covered. 
\par So we only have six product vectors in $\mathcal{S}_{\mathcal{U},4}$ given by
\begin{align*}
\ket{\psi_4} = \frac{1}{3} \left(\ket{0} + \ket{1} + \ket{2}\right) \left(\ket{0} + \ket{1} + \ket{2}\right), \qquad & \ket{11},\\
\frac{1}{\sqrt{2}} \ket{0} (\ket{0} + \ket{1}), \qquad& \frac{1}{\sqrt{2}} (\ket{0} + \ket{1})\ket{2},\\
\frac{1}{\sqrt{2}} \ket{2} (\ket{1} + \ket{2}), \qquad & \frac{1}{\sqrt{2}}  (\ket{1} + \ket{2})\ket{0}.
\end{align*}
\end{proof}
\subsection{Span of UPB in tripartite qubit system}\label{ssec2.3}


\par We consider the UPB $\mathcal{D}$ in $\mathbb{C}^2 \otimes \mathbb{C}^2 \otimes \mathbb{C}^2$ given in Equation (6) of Ref\cite{B2} and write the polynomial representation of this. For readability we first write the UPB which is as follows. 
\begin{alignat}{2}\label{upb3p}
\ket{\phi_1} &= \ket{0,1,+},\nonumber\\
\ket{\phi_2} &= \ket{1,+, 0}, \nonumber\\
\ket{\phi_3} &= \ket{+,0,1},\\
\ket{\phi_4} &= \ket{-,-,-},\nonumber 
\end{alignat} 
$\{\ket{0} ,\, \ket{1}\}$ are the standard basis of $\mathbb{C}^2$ and $ \ket{\pm} =\frac{\ket{0} \pm \ket{1}}{\sqrt{2}}.$ Let $\mathcal{V}$ be the span of all these vectors. We use the similar technique as earlier,  introducing the variable $Z$ to represent vectors of the third system. The polynomial representations of $\ket{\pm}$ in the first system is $\frac{1 \pm X}{\sqrt{2}}$. $X$ in the first system  is replaced by $Y$ and $Z$ in the second and the third system respectively. 
\par Thus by \cref{ss1.4} the polynomial representations of the vectors of \eqref{upb3p}, without normalisation scalars,  are given by 
\begin{alignat}{2}\label{upb3pp}
F_1(X,Y,Z) &= Y(1+Z),  \quad & F_2(X,Y,Z) &= X(1+Y), \nonumber\\
F_3(X,Y,Z) &=(1+X)Z,  	& F_4(X,Y,Z) &=(1-X)(1-Y)(1-Z).
\end{alignat}

Any non-zero linear combination of $\ket{\phi_j},\, 1 \le j \le 4$ other than their own scalar multiples has the form 
\[\ket{\psi} = \alpha \ket{\phi_1} + \beta \ket{\phi_2} + \gamma \ket{\phi_3} + \delta \ket{\phi_4},\]
with at least two of the scalars $\alpha, \beta, \, \gamma,\, \delta$  non-zero. 

\begin{theorem}\label{3qupb}
The linear span  $\mathcal{V}$ of the vectors in \eqref{upb3p} has only entangled vectors other than the scalar multiples of $\ket{\phi_j}$'s.  
\end{theorem}
\begin{proof}
The polynomial representation of $\ket{\psi}$ is given by 
\begin{eqnarray*}
\Psi(X,Y,Z) &=& \alpha Y(1+Z) + \beta X(1+Y) +\gamma (1+X)Z + \delta (1-X)(1-Y)(1-Z) \\
&=& [\alpha Y(1+Z) + \gamma Z + \delta (1-Y)(1-Z)] \\ &&+ X [\beta(1 + Y) + \gamma Z -\delta (1-Y)(1-Z)]\\
&=& [\delta +(\alpha - \delta) Y + (\gamma - \delta) Z + (\alpha + \delta) YZ] \\&& + X[ (\beta -\delta) + (\beta + \delta) Y + (\gamma + \delta) Z - \delta YZ] \\
&=& U(Y,Z) + X V(Y,Z), \quad \text{say,}
\end{eqnarray*}
where $U$ and $V$ are defined in an obvious way. Now $U=0$ if and only if $\delta =0 = \alpha =\gamma$. But at least two of the $\alpha, \, \beta, \, \gamma,\, \delta$ should be non-zero. So $U \neq 0$. Similarly $V \neq 0$ as well. 
\par Now we consider the case where both of them are non-zero. $\ket{\psi}$ is bipartite along $\{1\}$ vs $\{2,3\}$ cut, if and only if $V = a U$ for some non-zero scalar $a$. In that case 
\[\Psi(X,Y,Z) = (1 + a X) U(Y,Z).\]
This happens if and only if 
\begin{eqnarray}
\beta - \delta &=& a \delta,\\
\beta + \delta &=& a (\alpha - \delta),\\
\gamma + \delta &=& a (\gamma - \delta),\\
-\delta &=& a  (\alpha + \delta).
\end{eqnarray}
Simple calculations with the above equations give us that $\alpha =0$ which, in turn, gives that $\beta =0$ as well. As a result, the last equation gives us that $(a+1) \delta =0$. However $\delta$ can not be taken as zero, as by assumption at least two of the scalars $\alpha, \beta, \, \gamma,\, \delta$ are non-zero. Thus $a =-1$ which will imply that $\gamma =0$ a contradiction of the previous condition on the number of non-zero scalars. Therefore, the above system of equations is not consistent and hence $\ket{\psi}$ is not a product vector under the bipartite cut $\{1\}$ vs $\{2,3\}$, and a fortiori, not a product vector. 
\end{proof}
\begin{theorem}
The linear span $\mathcal{V}$ of $\{\ket{\phi_j}: 1 \le j \le 4\}$ has only genuinely entangled vectors other than the scalar multiples of $\ket{\phi_j},\, 1 \le j \le 4$. 
\end{theorem}
\begin{proof}
The role of $\ket{\phi_j}$'s for for $j=1,\, 2, \,3$ can be permuted by permuting the variables $X,\, Y,\, Z$. So the proof of the \cref{3qupb} can be re-modelled in other bipartite cuts as well. Hence $\ket{\psi}$ is genuinely entangled.  
\end{proof}
\subsection{Perturbation of $\mathcal{S}_\mathcal{V}=\mathcal{V}^\perp$ by a member of  UPB $\mathcal{D}$} 
\par Let $\mathcal{S}_\mathcal{V}$ be $\mathcal{V}^\perp$.
\begin{lemma}\label{lem2.2.3}
Let $\ket{\xi} \in \mathcal{S}_\mathcal{V}$ and $G(X,Y,Z)$ be its polynomial representation. Then $G=G_{a,b,c,d}$ with $(a,b,c,d) \in \mathbb{C}^4$  given by  
\begin{align}
G(X,Y,Z) =& d+2(a+b+c) +(aX +bY +cZ)\nonumber \\
&- (aXY +b YZ +c XZ) +d XYZ.
\end{align}
This gives rise to a linear bijective map between $\mathcal{S}_\mathcal{V}$ and $\mathbb{C}^4$. 
\end{lemma}
\begin{proof}
\par  $\mathcal{S}_\mathcal{V}= \{\ket{\phi_1}, \ket{\phi_2}, \ket{\phi_3}, \ket{\phi_4}\}^\perp$. Let $\ket{\xi} \in \mathcal{S_V}$ have the polinomial representation G(X,Y,Z). Then 
\newline $-$ coefficient of $YZ$ in $G = $ coefficient of $Y$ in $G=b$ (say).
\newline $-$ coefficient of $XY$ in $G = $ coefficient of $X$ in $G=a$ (say). 
\newline $-$ coefficient of $XZ$ in $G = $ coefficient of $Z$ in $G=c$ (say).
\newline Let $d'$ be the constant term. Then  $d' - ~($sum of the coefficients of $X,\, Y$, and $Z) ~+ ~($sum of the coefficients of $XY,\, YZ$, and $XZ) ~ - ~($coefficient of $XYZ)~=0$. So  
\[\text{coefficient of } XYZ =d' - 2(a+b+c).\]
As a consequence  
\[G(X,Y,Z)= d' + (aX + bY +cZ) -(aXY +bYZ +cXZ) +(d'-2(a+b+c)) XYZ.\]
 Putting  $d=d'- 2(a+b+c)$, we have 
\begin{align*}
G(X,Y,Z) &= d+2(a+b+c) +(aX +bY +cZ) - (aXY +b YZ +c XZ) +d XYZ.
\end{align*}
We may denote it by $G_{a,b,c,d}(X,Y,Z)$ with $(a,b,c,d) \in \mathbb{C}^4$. This map $(a,b,c,d) \mapsto G_{a,b,c,d}(X,Y,Z) \mapsto \ket{\xi_{a,b,c,d}}$ is linear and bijective. 
\end{proof}

\begin{lemma} \label{lem.2.2.4}
Let $\mathcal{S}_{\mathcal{V},4}$ be the linear span of $\mathcal{S_V}$ and $\ket{\phi_4}$. Then it coincides with the linear span of  $\mathcal{S_V}$  and $\ket{000}$, and also with that of $\mathcal{S_V}$  and $\ket{111}$.
\end{lemma}
\begin{proof}
Consider $\ket{\eta} = \ket{\phi_4} + \ket{\xi_{a,b,c,d}}$ with $(a,b,c,d) \in \mathbb{C}^4$. Its polynomial representation is 
\begin{align*}
K_{a,b,c,d}(X,Y,Z) &= [d+2(a+b+c) +1]+[(a-1)X +(b-1)Y + (c-1)Z]\\
& - [(a-1)XY +(b-1) YZ + (c-1) XZ] + (d-1) XYZ\\
&= 8 + G_{a-1,b-1,c-1, d-1}(X,Y,Z).
\end{align*}
So $\ket{000} \in \mathcal{S}_{\mathcal{V},4}$. We note that $\ket{\xi_{0,0,0,1}} \in \mathcal{S}_{\mathcal{V},4}$. So $\ket{111} \in\mathcal{S}_{\mathcal{V},4}$.
\end{proof}
\begin{theorem}\label{th2.9}
$\mathcal{S}_{\mathcal{V},4}$ is a quasi-completely entangled space with product index 6 with specified product vectors as given in the proof. 
\end{theorem}
\begin{proof}
It follows from the proof of \cref{lem2.2.3} that a general non-zero element $\ket{\xi}$ of $\mathcal{S}_{\mathcal{V},4}$ has polynomial representation  $F(X,Y,Z)$ of the  form 
\[ \alpha + (aX +bY +cZ) -a(XY + bYZ + cXZ)  + \beta XYZ,\]
with $\bm{0} \neq (\alpha, a,b,c,\beta)\in \mathbb{C}^5$. $\ket{\xi}$ is a product vector if and only if $F(X,Y,Z) = p(X)q(Y)r(Z)$ for some polynomials $p,\,q,\,r$. In this case, the lowest and the highest degree  homogeneous part of $F(X,Y,Z)$ have only one term. We express this condition as, \emph{$F$ is conical} following Ref \cite{MR3635750}. 
\par We clearly have six product vectors which correspond to 
\begin{eqnarray*}
(1,0,0,0,0), \quad &\text {viz. }& \ket{\xi} = \ket{000},\\
(0,0,0,0,1), \quad &\text {viz. }& \ket{\xi} = \ket{111},\\
(1,-1,-1,-1,-1), \quad &\text {viz. }& \ket{\xi} = (\ket{0}  - \ket{1}) \otimes (\ket{0}  - \ket{1}) \otimes (\ket{0}  - \ket{1}),\\
(0,1,0,0,0), \quad &\text {viz. }& \ket{\xi} = \ket{1} \otimes (\ket{0}  - \ket{1}) \otimes \ket{0} ,\\
(0,0,1,0,0), \quad &\text {viz. }& \ket{\xi} = \ket{0} \otimes \ket{1} \otimes (\ket{0}  - \ket{1})  ,\\
(0,0,0,1,0), \quad &\text {viz. }& \ket{\xi} = (\ket{0}  - \ket{1})\otimes \ket{0} \otimes \ket{1} .
\end{eqnarray*}
The normalised versions are obtained by replacing the last four by their suitable scalar multiples. 
It remains to show that there is no more product vector beyond those in the list. We confine our attention to the situation that $\ket{\xi}$ is a product vector. 
\par {\bf Case} $\alpha =0=\beta$ and $(a,b,c) \neq (0,0,0)$.  We have 
\[F(X,Y,Z) = aX +bY + cZ -(aXY + bYZ + cXZ).\]
As noted above, $F$ must be conical. So one and only one of the $a,\, b,\, c$ is not equal to zero, which corresponds to the last three elements of the above list. 
\par {\bf Case} $  \alpha \neq 0 = \beta$. Subcase $(a,b,c) \neq (0,0,0)$. But the product vectors have only conical polynomial representation. So one and only one of the $a,\, b,\, c$ is not equal to zero. Consider $a \neq 0 =b =c$. 
\[ F(X,Y,Z) = \alpha + aX -a XY,\]
which can not be of the form $p(X) q(Y)r(Z)$ for polynomials $p,\,q,$ and $r$. 
\par Similarly we can deal with other sub-cases and also with the {\bf case} $\alpha =0 \neq \beta$.
\par {\bf Case} $\alpha \neq 0 \neq \beta$.  
\begin{eqnarray*}
F(X,Y,Z) &=& (\alpha + bY + cZ -b YZ) \\
&& + X(a-aY -cZ +\beta YZ) \\
&=&U(Y,Z) + X V(Y, Z), \text{  say.}  
\end{eqnarray*}
Because of our assumption $U \neq 0$ and $V \neq 0$. So $\ket{\xi}$ is a product vector only if it is a product vector in the bipartite cut $\{\{1\},\{2,3\}\}$ if and only if $U = \gamma V$ for some scalar $\gamma \neq 0$. Suppose it is so. Then we have
\begin{align}
\alpha &= \gamma a, \tag{I}\\
b &= -\gamma a, \tag{II}\\
c &= -\gamma c, \tag{III}\\
-b &= \gamma \beta \tag{IV}.
\end{align}
We can observe that $a = \beta \neq0, \, b = -\alpha \neq 0$, and hence 
\begin{align*}
V(Y,Z) &= a -a Y -c Z +a YZ \\
&= a(1-Y)+(-c +aY)Z.
\end{align*}
So if $V( Y,Z) =q(Y) r(Z)$ then $c =a \neq 0$ and then $V(Y,Z) = a(1-Y)(1-Z)$. This with (III) gives $\gamma =-1$. So 
\[F(X,Y,Z)=a (1-X)(1-Y)(1-Z),\]
which corresponds to the third one in the list.
\end{proof}
\begin{theorem}\label{2.2.6}
The linear span $\mathcal{S}_{\mathcal{V},1}$ of  $\mathcal{S}_{\mathcal{V}}$ and $\ket{\phi_1}$ is QCES of product index   six  with specified product  vectors as given in the proof. 
\end{theorem}

\begin{proof}
It is enough to consider general elements in this space that have the polynomial representation of the form 
\begin{align*}
\Sigma(X,Y,Z) &=Y(1+Z) + d + 2(a+b+c) +  (aX +bY +cZ) \\
&-(aXY+bYZ +cXZ) +dXYZ \\
&= (d +2(a+b+c) +(1+b)Y +cZ +(1-b)YZ)\\
& +X(a -aY -cZ +d YZ) \\
&= U(Y,Z) +X V(Y,Z), \quad \text{ say with  } (a,b,c,d) \neq  \bm{0}.
\end{align*}
\par {\bf Case} $c \neq 0$. Then $U \neq 0 \neq V$. So $\Sigma(X,Y,Z)$ is of the form $p(X) Q(Y,Z)$ if and only if $U = \gamma V$ for some $\gamma \neq 0$. Suppose this does happen. Then  $\gamma$ can only be $-1$ as $cZ$ is present in $U$ and $-cZ$ is present in $V$. Simple calculations show that 
\begin{eqnarray*}
a &=& 1+b,\\
d &=& -1 +b,\\
c &=& -(3b+1) \neq 0.
\end{eqnarray*}
So $V(Y,Z) = (1+b) (1-Y) +((3b+1) +(b-1)Y)Z$, which has the form $q(Y) r(Z)$ if and only if $b=-1$ or $b=0$. Now $b=-1$ gives $a=0, \, c=2, \, d=-2$. Hence $V(Y,Z) =-2 (1+Y) Z$ and 
\[\Sigma (X,Y,Z) =2(1-X)(1+Y) Z.\]
\par On the other hand, $b=0$ gives $a=1,\, c=-1,\, d=-1$. So 
\[V(Y,Z) =1-Y + (1-Y)Z=(1-Y)(1+Z).\]
So $\Sigma(X,Y,Z)= (1-X)(1-Y)(1+Z)$.

\par Next, {\bf case} $c=0$. Then
\begin{align*}
U(Y,Z)&= d+2(a+b) +(1+b)Y+(1-b)YZ\\
V(Y,Z) &= a-aY +dYZ.
\end{align*}
Since $1-b$ and $1+b$ can not be both zero together we have $U(X,Z) \neq 0$. So $\Sigma (X,Y,Z)$ is of the form $P(X)Q(Y,Z)$ if and only if $V(Y,Z) = \gamma U(Y,Z)$ for some scalar $\gamma$. 
\par If $\gamma =0$, then $a =0=d$ and $b \neq 0$, and $U(Y,Z) =2b +(1+b)Y +(1-b)YZ$. This has the form $q(Y)r(Z)$ if and only if  $b=1$. (The other option $b=0$ has been discarded earlier.) And thus $U(Y,Z) =2 +2Y$. Hence 
\[\Sigma(X,Y,Z)= 2 +2Y.\]
\par Suppose $\gamma \neq 0$. Then 
\[\Sigma(Y,Z)= p(X)q(Y)r(Z)\]
only if $V(Y,Z) =q(Y)r(Z)$. Comparing the coefficients we get $a=0$ or $d=0$. Simple calculations give us possible product vectors as 
\[2(1+X)YZ, \quad \text{ and } \quad 2(1+X)(1-Y).\]
So altogether we get six vectors which, using the notations of \cref{upb3p}, are given as follows 
\begin{align*}
\ket{\phi_1} = \ket{0,1,+}, \qquad & \ket{-,+,1},\\
\ket{-,-,+}, \qquad& \ket{0,+,0},\\
\ket{+,1,1}, \qquad & \ket{+,-,0}.
\end{align*}
\end{proof}
\begin{rem}
We can get similar results by using other $\ket{\phi}$'s as well.
\end{rem}

\vspace{6pt}

 \section{Variants and perturbations of K. R. Parthasarathy's completely entangled subspace} \label{sec3}
\par We now develop methods for obtaining quasi-completely entangled perturbations of completely entangled subspaces given in Ref \cite{krp1} which has already been briefly introduced in \cref{ss1.2}. 
\par For the sake of convenience, we give the basics from Parthasarathy \cite{krp1} and a little variation by  Bhat \cite{bhat} in  a bit different notation, see also Ref \cite{PhysRevA.90.062323}. We follow the notation and terminology as in \cref{s1}  above. 
Let $N' = \sum_{j=1}^k d_j -k=N-1$.  Let  $\bm{i}=(i_j)_{j=1}^k,\,  0 \le i_j \le d_{j}-1$ for $1 \le j \le k,$ let $\mathcal{I}$ be  the collection  of all such topples  and for $0 \le n \le N',\, \mathcal{I}_n$, be the subset of $\mathcal{I}$ with $|\bm{i}|= \sum_{j=1}^k  i_j =n$
We note that $\{ \ket{\bm{i}}=e_{\bm{i}} = \bigotimes_{j=1}^k \ket{\bm{i}_j} \in \mathcal{I}\}$ is a basis for $\mathcal{H}$.
\subsection{Completely entangled subspace} \label{ss3.1}
We introduce completely entangled subspace by using \emph{van der Monde determinants} as given by Parthasarathy \cite{krp1}. Let
\begin{eqnarray}\label{eq2.1}
\ket{z_\lambda} &=& \bigotimes_{j=1}^k \sum_{s=0}^{d_j -1} \lambda^s \ket{s}, \quad \lambda \in \mathbb{C},\\
\ket{z_\infty} &=& \bigotimes_{j=1}^k \ket{d_j -1}. \nonumber
\end{eqnarray}
So, $\ket{z_\lambda} = \sum_{n=0}^{N'} \lambda^n\left( \sum_{\bm{i}: |\bm{i}|=n} \ket{\bm{i}} \right)$.  In particular, $\ket{z_0}= \ket{\bm{0}}$. 
\begin{enumerate}[(i)]
\item We call them {\bf van der Monde vectors} because for distinct $N$ complex number $\lambda$'s, the corresponding $\ket{z_\lambda}$'s are linearly independent because of van der Monde determinant as shown by Parthasarathy \cite{krp1}. 
\item Let $\mathcal{F}$ be the linear span of all van der Monde vectors, or, for that matter, any distinct $N$ van der Monde vectors. Then $\mathcal{F}^\perp$ is completely entangled as shown by Parthasarathy in Theorem 1.5 of Ref \cite{krp1}. We call it  {\bf K. R. Parthasarathy's space} or, simply {\bf Parthasarathy space}  and denote it {\bf $\mathcal{S}_P$}.
\end{enumerate}
\begin{rem}\label{rem3.1}
We note some related facts as given in Refs \cite{krp1,bhat,PhysRevA.90.062323} for further use. 
\begin{enumerate}[(i)]
\item  The set of product vectors  in $\mathcal{F}$ is exactly the set of all van der Monde vectors, viz., $\{ \ket{z_\lambda}: \lambda \in \mathbb{C}_\infty = \mathbb{C} \cup \{\infty\}\}$. 
\item The points $\ket{\eta}$  in the space $\mathcal{S}_P$ have the form $\sum_{n=0}^{N'} \sum_{\bm{i}:|\bm{i}|=n} a_{\bm{i}} \ket{\bm{i}} $, with $\sum_{\bm{i}:|\bm{i}|=n} a_{\bm{i}} =0$  for each $n,\, 0 \le n \le N'$.
\item As a consequence, $a_{\bm{0}}=0=a_{\bm{d}'}$ with $\bm{d}'$ having $d_{j }-1$ at $j$-th place and  for $0<n<N'$,  either all $a_{\bm{i}}$ with $|\bm{i}|=n$ are zero, or at least two of them are not zero.

\item For   $\lambda \in \mathbb{C}$ and $\ket{\eta} \in \mathcal{S}_P$, 
\[ \ket{z_\lambda} + \ket{\eta} = \ket{\bm{0}} + \sum_{n=1}^{N'-1} \sum_{|\bm{i}| =n} (\lambda^n + a_{\bm{i}}) \ket{\bm{i}} + \lambda^{N'} \ket{z_\infty}. \]
So 
\[\ket{z_0} + \ket{\eta}= \ket{\bm{0}} + \sum_{n=1}^{N'-1} \sum_{|\bm{i}| =n} a_{\bm{i}} \ket{\bm{i}}.\] 
Also 
\[\ket{z_\infty} + \ket{\eta}= \sum_{n=1}^{N'-1} \sum_{|\bm{i}| =n} a_{\bm{i}} \ket{\bm{i}} + \ket{\bm{d}'}. 
\]
\end{enumerate}
\end{rem}
It will be convenient to note the next obvious result as a lemma for further use. 

\begin{lemma}\label{lem3.1}
For $1 \le j \le k$, any non-zero vector $\ket{\xi_j}  \in \mathcal{H}_j$ has the form $\ket{\xi_j} = \sum_{t=m_j}^{M_j} b_{j t} \ket{t}$, for some $0 \le m_j \le M_j \le d_j -1$ with  $b_{j m_j} \neq 0 \neq b_{j M_j}$.  So 
\[\ket{\xi} = \bigotimes_{j=1}^k \ket{\xi_j} = \bigotimes_{j=1}^k \left( \sum_{t = m_{j}}^{M_j} b_{j t} \ket{t} \right) .\]

Let $\bm{m}=[m_j]_{j=1}^k, \, \bm{M}=[M_j]_{j=1}^k, \, |\bm{m}|=m$, and $|\bm{M}|=M$. Then $0 \le m \le M \le N'$ and $m =M$ if and only if $\bm{m} =\bm{M}$. Then 
\[\ket{\xi} =\prod_{j=1}^k b_j m_j \ket{\bm{m}},\]
if $\bm{m}=\bm{M}$, whereas 
\[\ket{\xi} = \prod_{j=1}^k b_j m_j \ket{\bm{m}} +  \sum_{m < n< M}  \sum_{|\bm{i}| =n}  \prod_{j=1}^k  b_{j i_j} \ket{\bm{i}} + \prod_{j=1}^k b_{j M_j} \ket{\bm{M}}\quad \text{ if } \bm{m} \neq \bm{M}.\]
We follow the convention that empty sums are corresponding  zeros and empty product of scalars is 1.
\end{lemma}

These expressions assume shorter forms  if we put  $\bm{b}_j$ for  $({b}_{j t})_{0 \le t \le d_j - 1}$, and then 
$\bm{b} = [\bm{b}_j]_{j=1}^k =[b_{jt}]$, and similar obvious forms.

\begin{defin}  
We let $\ket{\xi}_j$ etc.  to be as in \cref{lem3.1} above. 
If $b_j m_j = 1$ for $1 \le j \le k$, then we say that $\ket{\xi}$ is in \emph{neat form}  or $\ket{\xi}$ is a \emph{neat product  vector}.
\end{defin}

\par We  note that neat product  vectors give a unique representation of a non-zero  product vector treated as the one-dimensional spaces introduced  in \cref{s1}. Moreover,   van der Monde vectors  are neat product vectors. 
\subsection{Perturbation of $\mathcal{S}_P$ by a van der Monde vector.} 
We begin with a theorem  concerning $\ket{z_0}$ and $\ket{z_\infty}$ and then  go on to develop some \emph{Consistency conditions} for other van der Monde vectors.   We will find it convenient to look for neat product vectors   in these perturbed spaces.

\begin{theorem} \label{th:3.2}
Span $\mathcal{S}_{P, 0}$ of $\mathcal{S}_{P }$ and $\ket{z_0}$  contains no non-zero  product vector other than $\ket{z_0}$. Same is true for $\mathcal{S}_{P, \infty}$ with  $\ket{z_0}$ replaced by $\ket{z_\infty}$.
\end{theorem}
\begin{proof}
This follows directly after  combining \cref{lem3.1} with the observations in \cref{rem3.1} above. 
\end{proof}
\begin{eg}{$\mathbb{C}^3 \otimes \mathbb{C}^3$ system}
Let $\mathcal{H}_1=\mathcal{H}_2$ have dimension 3.. Then 
\[ \ket{z_0} = \ket{\bm{0}}, \quad
 \ket{z_\infty} = \ket{\bm{2}}. \]
If $0\neq \lambda \in \mathbb{C}$ then   $\ket{z_\lambda}$ and a non-zero element $\ket{\eta}\in \mathcal{S}_P$ have the form 
\begin{eqnarray*}
 \ket{z_\lambda} &=& \ket{00}  + \lambda( \ket{01} + \ket{10}) + \lambda^2 ( \ket{02} + \ket{11} + \ket{20}) +  \lambda^3 ( \ket{12} + \ket{21}) + \lambda^4\ket{22}),\\
\ket{\eta} &=& a(\ket{01} - \ket{10})  + (b \ket{02} -(b+c) \ket{11} + c \ket{20}) +d(\ket{12} - \ket{21}),
\end{eqnarray*}
with $a,\, b, \, c, \,d$ as scalars, not all zero.
\end{eg}
\begin{rem}[{\bf Consistency Conditions}]\label{rem3.2}
We now develop  some useful consistency conditions.   We will find it convenient to look for neat product vectors  in these perturbed spaces.
\begin{enumerate}[(i)]
\item Let $ \lambda \in \mathbb{C} \backslash \{0\}$. Suppose $ \ket{z_\lambda} + \ket{\eta} = \ket{\xi} $ is a product vector as in \cref{lem3.1} above for some  $0 \neq \ket{\eta} \in \mathcal{S}_P$ as in \cref{rem3.1}.  Then 
\begin{enumerate}[(a)]
\item  $\bm{m}= \bm{0}, \, \bm{M} = \bm{d}',\, b_{\bm{0}}=  \prod_{j=1}^k b_{j \bm{0}} =1, \, b_{\bm{d}'}= \prod_{j=1}^k b_{j ,d_j -1} =\lambda^{N'}$, \label{rem3.2i}
\item for $1 \le n \le N' -1, \, |\bm{i}| =n, \, \lambda^n + a_{\bm{i}} =b_{\bm{i}}= \prod_{j=1}^k b_{j, i_j}$, and 
\[\sum_{\bm{i}:|\bm{i}|=n} b_{\bm{i}}=\sum_{\bm{i}:|\bm{i}| =n}  \prod_{j=1}^k b_{j, i_j} = \lambda^n \# \mathcal{I}_n \neq 0.\] \label{rem3.2ii}
\item For some  $1 \le n \le N' - 1$, $\bm{i}$ with $|\bm{i}|=n,\,  b_{\bm{i}}\neq  \lambda^n$.  \label{rem3.2iii}
\end{enumerate} 
On the other hand, if these conditions are satisfied  
 then we have a product vector in $\mathcal{S}_{P,\lambda}$ other than $\ket{z_\lambda}$. 
 We may try to figure out solutions $\bm{b}$. 
\par Hence we call the above conditions  (a) to (c) as \emph{consistency conditions} for $\ket{z_\lambda} + \ket{\eta}$ to be a product vector for some $0 \neq \ket{\eta} \in \mathcal{S}_P$. 

\item We next note that  $\bm{b} =[b_{jt}]= [\lambda^t c_{jt}]$
is a solution for $\lambda \in \mathbb{C} \backslash \{0\}$, if and only if $\bm{c} = (c_{j t})_{j=1}^k$, 
is a solution for $\lambda =1$. So it is enough to consider the case $\lambda =1$. 
\item We display the Consistency conditions for $\lambda=1$ for practical purposes.
\begin{enumerate}
\item We may take    $c_{j0} =1,\, 1 \le j \le k$ to proceed.  In other words, we look  for product vectors in their neat forms. 
\item $\prod_{j=1}^k c_{j ,d_j-1} =1$. 
\item  $\sum_{\bm{i}: |\bm{i}| = n} \prod_{j=1}^k c_{j, i_j} = \# \mathcal{I}_n, \, 1 \le n \le N'-1.$
\item For some  $1 \le n \le N' - 1$, some $\bm{i}$ with $|\bm{i}|=n,\,  c_{\bm{i}}\neq  1$.
\end{enumerate}
\item We have $N'$ scalars $(c_{j, p})_{\substack{1 \le j \le k\\ 1 \le p \le d_j -1}},$ to be determined from effectively  $N'$ consistency conditions. To give an idea, we note them  in the following specific cases available for suitably large  $N' >n$ and dimensions $d_j$'s. 
\begin{enumerate}
\item[Case $n=1$:] $ \sum_{j=1}^k c_{j,1} =k$. 
\item[Case $n=2$:]  $ \sum_{1 \le j \neq j'  \le k} c_{j1}c_{j'1} + \sum_{1 \le j \le k} c_{j2} = \binom{k}{2} + k= \frac{k(k+1)}{2}$ if for all $j$, $d_j>2$.
\end{enumerate}
We may recall here that empty sums are zero and empty products are 1.
\begin{enumerate}
\item[Case $n=3$:]  Furthermore,  $\sum_{1 \le j \neq j' \neq j'' \neq j\le k } c_{j1}c_{j'1}c_{j''1} + \sum_{1 \le j \neq j'  \le k} c_{j2}c_{j'1} + \sum_{j=1}^k c_{j3} = \binom{k}{3} + 2 \binom{k}{2} +k,$ 
 if $k\ge 3$ and $d_j \ge 4$ for all $j$, and so on. 
\end{enumerate}
\item For $k=2=d_1=d_2$, consistency conditions have  no solution. So the span $\mathcal{S}_{P, \lambda}$ of $\ket{z_\lambda}$ and $\mathcal{S}_P$ is  for $0 \neq \lambda \in \mathbb{C}$,   is QCES with product index 1.
\item We have not been able to prove or disprove consistency conditions in general. However, the situation is clear for some special categories and specific examples of two quitrits and three qubits to allow comparison with results in \cref{sec2}  above. 
\end{enumerate}
\end{rem}
\subsection{Example of multipartite qubit  systems}
We consider the case of all qubits, i.e., for $1 \le j \le k, \,d_j=2$. We have the following generalisation of \cref{rem3.2} (v) for the case $k=2$. For $k=3$ this gives a contrast with  \cref{th2.9} and \cref{2.2.6}.
\begin{theorem}\label{th:3.6} 
Let $d_j=2$ for $1 \le j \le k$ and $0 \neq \lambda \in \mathbb{C}$. Then the span $\mathcal{S}_{P,\lambda}$ of $\ket{z_\lambda}$ and $\mathcal{S}_P$ is QCES with product index $1$. 
\end{theorem}
\begin{proof} 
In view of \cref{rem3.2}(ii), it is enough to consider the case of $\lambda=1$. We first note that $N'=k$,     for $\bm{i} \in \mathcal{I}$ and  $1 \le j \le k$,   $i_j =0$ or $1$, and $  \#\mathcal{I}_n = \binom{k}{n}$ for $1 \le n \le k-1$.   
So using consistency relations in \cref{rem3.2}, $c_{j1}$'s are roots of $(x-1)^k=0$ and, therefore,  are all $1$. Thus $\ket{\zeta} \in \mathcal{S}_{\mathcal{P},1}$ is a product vector if and only if $\ket{\zeta} = \ket{z_1}$. Hence we do not get any new product vectors. 
\end{proof}

\subsection{Bipartite set-up}

\par  Let $k=2,\, 2 \le  d_1 , d_2 \le d$, say. It is enough to consider the case $d_1 \le d_2 =d$. 
So $2 \le N' \le2d-2$. Consistency  Conditions  take manageable simple forms. We first consider the case $2=d_1 < d_2$. We obtain the following Theorem.

\begin{theorem}\label{th3.5}
Let $d_1=2 < d =d_2$. Then the span $\mathcal{S}_{P,1}$ of $\ket{z_1}$ and $\mathcal{S}_P$ is QCES of product index $d$ if $d$ is odd, but $d-1$ for $d$ even, with specified product vectors. 
\end{theorem}
\begin{proof}
We have $N'=d$.  Also for $1 \le n \le d-1,\, \mathcal{I}_n= \{(0,n),(1,n-1)\}$. The consistence conditions in \cref{rem3.2} take simpler forms, more so for neat form,  which only we consider, 
and thus  have $c_{10}=c_{20}=1$.  
For notational convenience, let $\alpha =c_{11}$ and for $0 \le p \le d-1, \, c_p=c_{2p}$. Then

 {\bf Case} $n=N'=d$. Here  $\alpha c_{d-1} =1$. So $\alpha \neq 0, c_{d-1} =\frac{1}{\alpha}$.

{\bf Case} $1 \le n \le N' -1 =d-1$. Note that  $\alpha c_{n-1} + c_n=2.$ We first note that $\alpha=1$ leads to $c_n=1$ for all $1 \le n \le d-1$, which does not satisfy \cref{rem3.2} (d). So we discard it.  We multiply the equation for $r$  by  $(-\alpha)^{d-1-r} ,\,  1 \le r \le d -1$ and obtain the system below.  
\begin{eqnarray*}
\alpha (-\alpha)^{d-2} c_0 + (-\alpha)^{d-2} c_1 &=& 2  (-\alpha)^{d-2},\\
\alpha (-\alpha)^{d-3} c_1 + (-\alpha)^{d-3} c_2 &=& 2  (-\alpha)^{d-3},\\
\vdots && \vdots\\
\alpha (-\alpha)^{d-1-r} c_{r-1} + (-\alpha)^{d-1-r} c_r &=& 2  (-\alpha)^{d-1-r},\\
\vdots && \vdots\\
\alpha (-\alpha) c_{d-3} + (-\alpha) c_{d-2} &=& 2  (-\alpha),\\
\alpha  c_{d-2} +  c_{d-1} &=& 2 .
\end{eqnarray*} 
So after adding all and cancelling, we get 
\[ \alpha (-\alpha)^{d-2} + \frac{1}{\alpha} =2[(-\alpha)^{d-2} + \cdots +(-\alpha)^{d-1-r} +\cdots+ 1].\]
Hence 
\[(-\alpha)^d +1 = 2 \alpha[(-\alpha)^{d-2} + \cdots +(-\alpha)^{d-1-r} +\cdots+ 1].\]
This equation is not consistent for $\alpha =-1$ as $ 2 \neq -2 (d-1)$. So $\alpha \neq -1$. For $\alpha \neq -1$ we have  
\[ (-\alpha)^{d} + 1 = 2 \alpha \frac{1- (-\alpha)^{d-1}}{1+\alpha}.\]
Simplifying and solving the equation we get either $\alpha =1$ or $(-\alpha)^d =1$. We have already set aside the case of $\alpha=1$. Hence the solution is   
\begin{eqnarray*}
\alpha_t &=& -e^{\frac{2 \pi t \imath}{d}}, \qquad 1 \le t \le d-1, \; t \neq \frac{d}{2}.
\end{eqnarray*}
So the number of relevant solution is $d-1$ if $d$ is odd, and $d-2$ if $d$ is even.  
\newline Let $ 1 \le r \le d-1$.
Adding the system up to $r$ we get 
\[ \alpha (-\alpha)^{d-2}  + (-\alpha)^{d-1-r} c_r = 2 [ (-\alpha)^{d-2} + \cdots + (-\alpha)^{d-1 -r}].\]
After simplifying we get 
 \[c_r =\frac{1}{1+\alpha} [2 + (-1)^r \alpha^r (\alpha -1)].\]
We have  new neat product vectors corresponding to $\alpha_t = - \exp \left[\frac{2\pi \imath t}{d} \right]$ for $1 \le t \le d-1$, $t \neq \frac{d}{2}$, and  their number is $(d-1)$ if $d$ is odd and $d-2$ if  $d$ is even. Hence  $\mathcal{S}_{P, 1}$ contains all together (including $\ket{z_1}$) exactly $d$ specified product vectors  if $d$ is odd, and $d-1$ specified  product vectors if $d$ is even.  
\end{proof}
\begin{theorem}
Let $d_1=2 < d =d_2$. For any $\lambda \in \mathbb{C}\backslash \{0\}$ the span $\mathcal{S}_{P, \lambda}$ of $\ket{z_\lambda}$ and $\mathcal{S}_P$ contains $d$ specified product vectors (including $\ket{z_\lambda}$) if $d$ is odd and $d-1$ specified product vectors (including $\ket{z_\lambda}$) if $d$ is even.
\end{theorem}
\begin{proof}
The proof follows by modifying the proof of \cref{th3.5} above using  \cref{rem3.2}, particularly, part (ii).
\end{proof}
\begin{rem}\label{rem3.3}
 Let $k =2, \, d_1 =d_2 =d \ge 3$. Then $N'=2d-2 \ge 4.$ 
\begin{enumerate}[(i)]
\item Consistency conditions in \cref{rem3.2} take somple form because 
\[ \mathcal{I}_n  = \{ (p, n-p): 0 \le p \le n\}, \quad \text{ for  } 1 \le n \le d-1,\]
and 
\[ \mathcal{I}_n  = \{ (p, n-p): n-d+1 \le p \le d-1\}, \quad \text{ for  } d \le n \le 2d-3.\]
As a consequence the cardinality 
\[ \#\mathcal{I}_n = \begin{cases}
n+1 & \text{ for } 1 \le n \le d-1, \text{ and }\\
2d -n -1 & \text{ for } d \le n \le 2d-3.
\end{cases}\]
The construction can be seen in \cref{fig:1}.

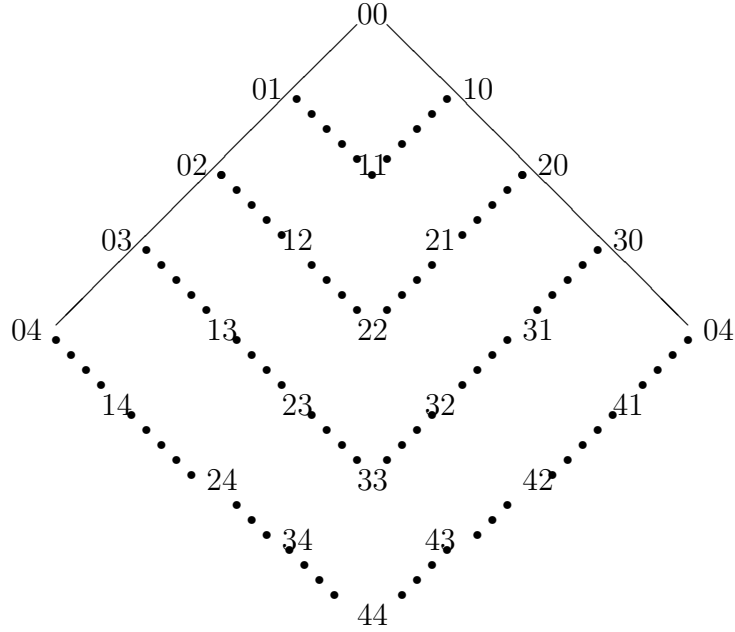
\begin{figure}[!h]
\begin{center}
\setlength{\unitlength}{2mm}
\begin{picture}(40,40)(-20,-20)
\put(1,20){\line(1,-1){20}}
\put(-1,20){\line(-1,-1){20}}
\put(-1,20){$00$}

\multiput(-10,10)(10,0){3}{\circle*{0.5}}

\multiput(-5,15)(1,-1){5}{\circle*{0.5}}\put(-8,15){$01$}
\multiput(5,15)(-1,-1){5}{\circle*{0.5}}\put(6,15){$10$} 
\put(-1,10){$11$}
\multiput(-10,10)(1,-1){5}{\circle*{0.5}}\put(-13,10){$02$}
\multiput(10,10)(-1,-1){5}{\circle*{0.5}}\put(11,10){$20$}
\multiput(-4,4)(1,-1){4}{\circle*{0.5}}\put(-6,5){$12$}
\multiput(4,4)(-1,-1){4}{\circle*{0.5}}\put(3.5,5){$21$}
\put(-1,-1){$22$}
\multiput(-15,5)(1,-1){5}{\circle*{0.5}}\put(-18,5){$03$}
\multiput(15,5)(-1,-1){5}{\circle*{0.5}}\put(16,5){$30$}
\put(-11,-1){$13$}
\put(10,-1){$31$}
\multiput(-9,-1)(1,-1){4}{\circle*{0.5}}
\multiput(9,-1)(-1,-1){4}{\circle*{0.5}}
\multiput(-4,-6)(1,-1){4}{\circle*{0.5}}\put(-6,-6){$23$}
\multiput(4,-6)(-1,-1){4}{\circle*{0.5}}\put(3.5,-6){$32$}
\put(-1,-11){$33$}
\put(-24,-1){$04$}
\put(22,-1){$04$}
\put(-18,-6){$14$}
\put(16,-6){$41$}
\put(-11,-11){$24$}
\put(10,-11){$42$}
\put(-6,-15){$34$}
\put(3.5,-15){$43$}
\multiput(-21,-1)(1,-1){4}{\circle*{0.5}}
\multiput(21,-1)(-1,-1){4}{\circle*{0.5}}
\multiput(-16,-6)(1,-1){5}{\circle*{0.5}} \multiput(16,-6)(-1,-1){5}{\circle*{0.5}}
\multiput(-9,-12)(1,-1){3}{\circle*{0.5}}\multiput(9,-12)(-1,-1){3}{\circle*{0.5}}
\multiput(-5.5,-15)(1,-1){4}{\circle*{0.5}}
\multiput(5,-15)(-1,-1){4}{\circle*{0.5}}
\put(-1,-20){$44$}
\end{picture}
\end{center}
\caption{Description of $\mathcal{I}_n$'s for  $ 2 \le d \le 5$. }\label{fig:1} 
\end{figure}
\item We let $a_s =c_{1s},\, b_t = c_{2t}$, for $0 \le s,\, t \le d-1$. Then we can write $a_s b_t$ in the place  $s t$ keeping in mind that $a_0=1=b_0$.  
 A pictorial representation is given in \cref{fig:2}. For $0 \le n \le N'$, such products for $\mathcal{I}_n$ are in  the $n$-th row in \cref{fig:2}. The consistency conditions are expressed in \cref{fig:2}  by specifying the sum of $n$th row as the number of entries in that row. We note this as follows
 \begin{itemize}
 \item $a_0=1=b_0$, in other words, we   concentrate on product vectors in their neat forms.
 \item For $1 \le n \le d-1, \,\sum_{\substack{0 \le s,t \le d-1 \\ s+t =n}} a_s b_t =n+1.$
 \item For $ d \le n \le 2d -3, \, \sum_{\substack{0 \le s,t \le d-1 \\ s+t =n}} a_s b_t =2d-n-1.$
 \item $a_{d-1}b_{d-1} =1$.
 \item At least one of the entries occurring in the left hand side above is not equal to 1. Note, as a consequence, at least one more in the same row is not equal to $1$. 
 \end{itemize}

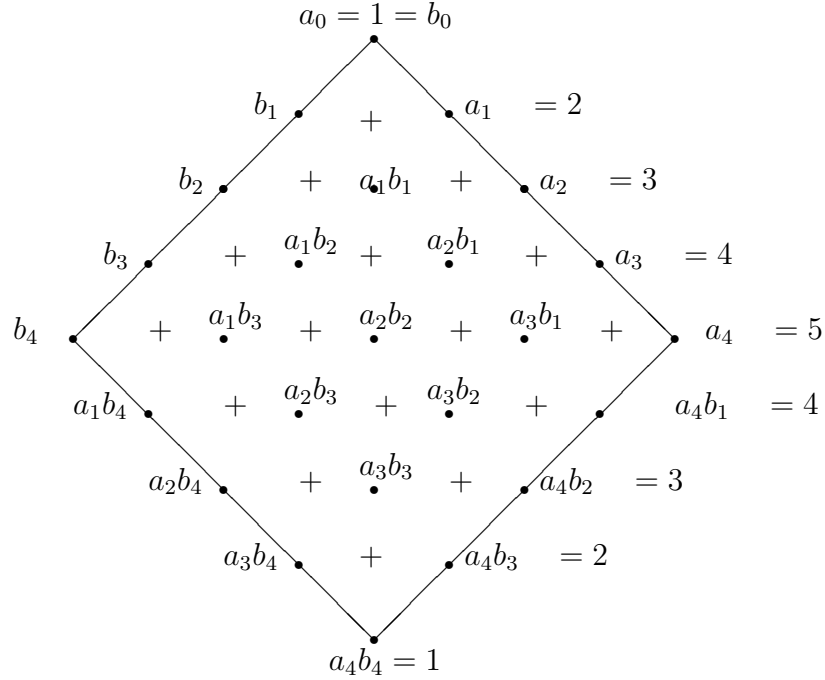
\begin{figure}[!h]
\begin{center}
\setlength{\unitlength}{2mm}
\begin{picture}(40,50)(-20,-25)
\put(0,20){\line(1,-1){20}}
\put(0,20){\line(-1,-1){20}}
\put(-5,21){$a_0=1=b_0$}

\put(0,20){\circle*{0.5}}
\put(0,-20){\circle*{0.5}}
\put(-20,0){\line(1,-1){20}}
\put(20,0){\line(-1,-1){20}}
\put(-5,15){\circle*{0.5}}\put(-8,15){$b_1$}
\put(5,15){\circle*{0.5}}\put(6,15){$a_1 \quad =2$} 
\put(-1,14){$+$}

\put(-10,10){\circle*{0.5}}	\put(-13,10){$b_2$}	\put(-1,10){$a_1b_1$} 
\put(10,10){\circle*{0.5}}	\put(11,10){$a_2 \quad =3$} 	\put(-5,10){$+$}\put(5,10){$+$}
\multiput(-15,5)(10,0){4}{\circle*{0.5}}\put(-18,5){$b_3$}

\put(16,5){$a_3 \quad =4$}

\put(-6,6){$a_1b_2$}

\put(3.5,6){$a_2b_1$}
\put(-10,5){$+$} \put(10,5){$+$} \put(-1,5){$+$}
\put(-24,0){$b_4$} 
\multiput(-20,0)(10,0){5}{\circle*{0.5}}
\put(22,0){$a_4 \quad =5$} 

\put(-11,1){$a_1b_3$}

\put(-1,1){$a_2b_2$}

\put(9,1){$a_3b_1$}
\multiput(-15,0)(10,0){4}{$+$}
\multiput(-10,10)(10,0){3}{\circle*{0.5}}

\multiput(-15,-5)(10,0){4}{\circle*{0.5}}
\put(-20,-5){$a_1b_4$}
\put(-6,-4){$a_2b_3$}
\put(3.5,-4){$a_3b_2$}
\put(20,-5){$a_4b_1 \quad =4$}
\multiput(-10,-5)(10,0){3}{$+$} 
\multiput(-10,-10)(10,0){3}{\circle*{0.5}}
\put(-15,-10){$a_2b_4$}
\put(11,-10){$a_4b_2 \quad =3$}
\put(-1,-9){$a_3b_3$}
\multiput(-5,-10)(10,0){2}{$+$}
\put(-5,-15){\circle*{0.5}}\put(-10,-15){$a_3b_4$}
\put(5,-15){\circle*{0.5}}\put(6,-15){$a_4b_3 \quad =2$} 
\put(-1,-15){$+$}

\put(-3,-22){$a_4b_4=1$}
\put(-20,-25){At least one of the terms is not equal to 1.}
\end{picture}
\end{center}
\caption{Consistency conditions diamond for $d=5$.  }\label{fig:2} 
\end{figure}
\noindent We note the consistency conditions for $d=3$ and $d=4$.
\item {\bf Case} of $d=3$. 
\begin{enumerate}[(a)]
\item $a_0 =1 =b_0$.
\begin{align*}
n &=1 :& &b_1 +a_1 =2,\\
n &=2 :& &b_2 +a_1 b_1 +a_2  =3,\\
n &=3 :& &a_1b_2 +a_2 b_1=2,\\
 && & a_2 b_2 =1.
\end{align*}
At least one of the terms in the rows $n=1, \, 2, \, 3$ is not equal to $1$. 
\item Suppose conditions given in (a) are all satisfied. If $a_1=1$, then $b_1=1$ and so $a_2 +b_2 =2$ and altogether $a_2=b_2=1$. Indeed we can show that setting any of $b_1, \, a_2,\,  b_2$ is $1$ forces all $a$'s and $b$'s to be $1$. As this is not permitted, we may declare that none of the $a_1,\, b_1,\, a_2,\, b_2$ is $1$. 
\end{enumerate}
\item {\bf Case} $d=4$. 
\begin{enumerate}[(a)]
\item $a_0 =1 =b_0$. 
\begin{align*}
n &=1 :& &b_1 +a_1 =2,\\
n &=2 :& &b_2 +a_1 b_1 +a_2  =3,\\
n &=3 :& &b_3+ a_1b_2 +a_2 b_1+a_3=4,\\
n &=4:& & a_1b_3+ a_2 b_2 +a_3b_1=3,\\
n &=5:& &  a_2 b_3 +a_3b_2=2,\\
&&& a_3 b_3 =1.
\end{align*}
At least one of the terms of the left hand side is not equal to 1. 
\item Suppose conditions in (a) above are all satisfied. Suppose $a_1=1$. Then we have $b_1=1$. So $b_2 +a_2 =2$. Further $b_3 +b_2 +a_2 +a_3 =4$ giving us $b_3 + a_3 =2$. Next condition gives us $a_2b_2 =1$. Hence the last line give $a_3 =1=b_3$ and also $a_2 =b_2=1$, which is not permitted. So $a_1 \neq 1$. Similarly we can show that none of $b_1, \, b_3,\, a_3$ is equal to 1.
\end{enumerate}

\item 
Let $ 1 \le s, \,t \le d-1.$ Put
\begin{align*}
\alpha_s &= \frac{1}{2} (a_s + b_s), &\beta_s &= \frac{1}{2} (a_s - b_s), \quad \text{ which gives } \\
a_s &= \alpha_s + \beta_s, & b_s &= \alpha_s - \beta_s.
\end{align*} 
Simple calculations show  that 
\[a_s b_t + b_s a_t =2( \alpha_s \alpha_t - \beta_s \beta_t). \]
In particular $a_sb_s = \alpha_s^2 - \beta_s^2$. 
\end{enumerate}
\end{rem}
Thus, we can have an alternative system of consistency conditions. We utilise it in the proof of our next theorem.
\begin{theorem}\label{th3.6}
For the space $\mathbb{C}^3 \otimes \mathbb{C}^3$, the span $\mathcal{S}_{P,1}$ of $\mathcal{S}_{P}$ and $\ket{z_1}$ has product index 3 with specified product vectors. 
\end{theorem}
\begin{proof}
 The consistency conditions in \cref{rem3.3} (v) above take the form $\alpha_1 =1, \, \alpha_2 = 1 + \frac{1}{2} \beta_1^2$ together with,  
\begin{align*}
&\alpha_2 = \beta_1\beta_2 +1, \\
&\alpha_2^2 - \beta_2^2 = 1.  
\end{align*} 
For solutions other than the obvious $\ket{z_1}$, indeed the system has two solutions which are:
\begin{align*}
\beta_1&= \sqrt{3} \imath, & \beta_2 &= \frac{\sqrt{3}}{2} \imath, & \alpha_2&= - \frac{1}{2}, \quad \text{ and }\\
\beta_1&=- \sqrt{3} \imath, & \beta_2 &= -\frac{\sqrt{3}}{2} \imath, & \alpha_2&= - \frac{1}{2}.
\end{align*}
More specifically, in terms of the cube roots $1, \, \omega, \, \omega^2$ of unity, 
\begin{align*}
a_1&= -2 \omega^2, & b_1&= -2 \omega, & a_2 &= \omega, & b_2 &= \omega^2,  \quad \text{ and }\\
a_1&= -2 \omega, & b_1&= -2 \omega^2, & a_2 &= \omega^2, & b_2 &= \omega. 
\end{align*}
This gives two new product vectors 
\begin{eqnarray*}
&& \left(\ket{0} -2 \omega^2 \ket{1} + \omega \ket{2} \right) \otimes \left(\ket{0} -2 \omega \ket{1} + \omega^2 \ket{2} \right),\\
&& \left(\ket{0} -2 \omega \ket{1} + \omega^2 \ket{2} \right) \otimes \left(\ket{0} -2 \omega^2 \ket{1} + \omega \ket{2} \right).
\end{eqnarray*}
in the span $\mathcal{S}_{P,1}$ of $\mathcal{S}_{P}$ and $\ket{z_1}$.  
\end{proof}
\begin{theorem}
For $\mathbb{C}^3 \otimes \mathbb{C}^3$ system,  $0 \neq \lambda \in \mathbb{C}$, the linear span $\mathcal{S}_{P,\lambda}$ of $\mathcal{S}_P$ and  $\ket{z_\lambda}$ has product index 3 with specified product vectors. 
\end{theorem}
\begin{proof}
It follows from \cref{th3.6} in view of \cref{rem3.2} (ii). 
\end{proof}
\begin{rem}
We note that a solution for consistency conditions for any particular $d$ can not be used to give a solution for $d+1$ just by trying to suitably solve for the missing terms and vice versa.  For instance, for $d=3$, this leads to $ b_3 +a_3 =2,\, a_1b_3 +a_3b_1 =2,\, a_2 b_3 + a_3 b_2 =2$, and $a_3 b_3=1$, which  gives the values $a_3 =1 =b_3,\, a_1 +b_1 =2,\, a_2 +b_2 =2$, and $a_2 b_2 =1$. This is not true. This can be seen from \cref{fig:3}.  
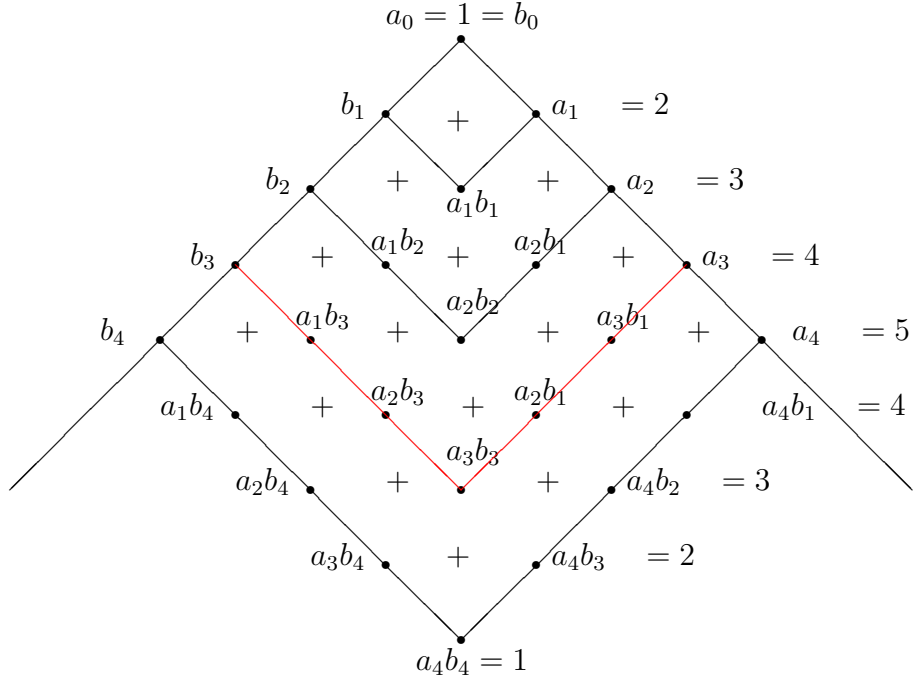
\begin{figure}[!h]
\begin{center}
\setlength{\unitlength}{2mm}
\begin{picture}(44,44)(-22,-22)
\put(0,20){\line(1,-1){30}}
\put(0,20){\line(-1,-1){30}}
\put(-5,21){$a_0=1=b_0$}
\put(0,20){\circle*{0.5}}
\put(0,-20){\circle*{0.5}}

\put(-5,15){\circle*{0.5}}\put(-8,15){$b_1$}
\put(5,15){\circle*{0.5}}\put(6,15){$a_1 \quad =2$}

\multiput(-10,10)(10,0){3}{\circle*{0.5}}	\put(-13,10){$b_2$}	\put(-1,8.5){$a_1b_1$} 
\put(11,10){$a_2 \quad =3$} 	
\multiput(-15,5)(10,0){4}{\circle*{0.5}}\put(-18,5){$b_3$}

\put(16,5){$a_3 \quad =4$}

\put(-6,6){$a_1b_2$}

\put(3.5,6){$a_2b_1$}

\put(-24,0){$b_4$} 
\multiput(-20,0)(10,0){5}{\circle*{0.5}}
\put(22,0){$a_4 \quad =5$} 

\put(-11,1){$a_1b_3$}

\put(-1,2){$a_2b_2$}

\put(9,1){$a_3b_1$}

\multiput(-15,-5)(10,0){4}{\circle*{0.5}}
\put(-20,-5){$a_1b_4$}
\put(-6,-4){$a_2b_3$}
\put(3.5,-4){$a_2b_1$}
\put(20,-5){$a_4b_1 \quad =4$}

\multiput(-10,-10)(10,0){3}{\circle*{0.5}}
\put(-15,-10){$a_2b_4$}
\put(11,-10){$a_4b_2 \quad =3$}
\put(-1,-8){$a_3b_3$}

\put(-5,-15){\circle*{0.5}}\put(-10,-15){$a_3b_4$}
\put(5,-15){\circle*{0.5}}\put(6,-15){$a_4b_3 \quad =2$}

\put(-3,-22){$a_4b_4=1$}
\put(-5,15){\line(1,-1){5}}
\put(5,15){\line(-1,-1){5}}
\put(-10,10){\line(1,-1){10}}
\put(10,10){\line(-1,-1){10}}
{\color{red}\put(-15,5){\line(1,-1){15}}
\put(15,5){\line(-1,-1){15}}}
\put(-20,0){\line(1,-1){20}}
\put(20,0){\line(-1,-1){20}}
\put(-1,14){$+$}
\put(-5,10){$+$}\put(5,10){$+$}
\put(-10,5){$+$} \put(10,5){$+$} \put(-1,5){$+$}
\multiput(-15,0)(10,0){4}{$+$}
\multiput(-10,-5)(10,0){3}{$+$} 
\multiput(-5,-10)(10,0){2}{$+$}
\put(-1,-15){$+$}

\end{picture}
\end{center}
\caption{Non-compatibility of consecutive consistency conditions.}\label{fig:3} 
\end{figure}

\end{rem}

\section{Examples of double perturbations with infinitely many pure product states}

In this section we consider double perturbations of completely entangled subspaces of multipartite systems that have infinitely many  pure product states; in other words, they are not QCES. We only give some instances rather than the complete picture.  We divide the discussion in two parts corresponding to sections,  \ref{sec2} and \ref{sec3} above. 
\subsection{UPB related set-up} We first consider the situation for $3 \otimes 3$ system as in subsections \ref{ssec2.1} and \ref{ssec2.2} above. 
\begin{theorem}
The span $\mathcal{S}_{\mathcal{U},0,1}$ of  $\mathcal{S_U}, \, \ket{\psi_0},$ and $\ket{\psi_1}$ has infinitely many product vectors, in other words, it is not QCES.
\end{theorem}
\begin{proof}
We refer to \cref{rem2.2}, particularly part (iii) above with $\bm{\gamma}$ to be suitably chosen. Let $0\neq t \in \mathbb{C}$. Corresponding to $\bm{\gamma} =(1,t,0,0)$, we have 
\[\Sigma_t(X,Y) =G(X,Y) -(1-Y) -t(X^2Y^2 -X^2Y)\]
is in the linear span $\mathcal{S}_{\mathcal{U},0,1}$ of  $\mathcal{S}_{\mathcal{U}}, \, \ket{\psi_0}$ and $\ket{\psi_1}$.  The above equation takes the form 
\begin{eqnarray*}
\Sigma_t(X,Y)  &=& 2Y +2tX^2Y -2(t+1) XY\\
&=& 2(1-(t+1) X +tX^2)Y,
\end{eqnarray*}
which is a  product vector. Thus $\mathcal{S}_{\mathcal{U},0,1}$ has infinitely many product vectors.
\end{proof}
\begin{theorem}
 The span $\mathcal{S}_{\mathcal{U}, 0,4}$ of $\mathcal{S_U}, \, \ket{\psi_0}$ and $\ket{\psi_4}$ has infinitely many pure product states. In other words, it is not QCES. 
\end{theorem}
\begin{proof}
 We refer to \cref{th2.4}.  Let $t \in \mathbb{C}$. Consider $(1+t)(1+Y)$. It is in $\mathcal{S}_{\mathcal{U},4}$ by taking $\tilde{\bm{\gamma}}$ with $\gamma_0=1+t, \,\gamma_j=0,\,j=1,\,2,\,3,\,4$. Also $(1-t)(1-Y)$ is a scalar multiple of $\ket{\psi_0}$. So 
\[\Sigma_t(X,Y) = (1-t)(1-Y) + (1+t)(1+Y)\]
is in  $\mathcal{S}_{\mathcal{U},0,4}$. But 
\[\Sigma_t(X,Y) = 2+2tY =2(1+tY)\]
gives a product vector. So $\mathcal{S}_{\mathcal{U},0,4}$ has infinitely many product vectors of the form 
\[\frac{1}{\sqrt{1+|t|^2}}\ket{0} \otimes (\ket{0} + t \ket{1})\] 
for any $t \in \mathbb{C} \backslash \{0\}$.  
\end{proof}
\par We now come to the tripartite qubit UPB as in \cref{ssec2.3}. We state and prove a double perturbation theorem on the linear  span $\mathcal{S}_{\mathcal{V}, 1,2}$ of $\mathcal{S_V}, \,\ket{\phi_1}$ and $\ket{\phi_2}$.
\begin{theorem}
The span $\mathcal{S}_{\mathcal{V},1,2}$ contains infinitely many product vectors.
\end{theorem}
\begin{proof}
Let $0 \neq \beta \in \mathbb{C}$ and $\bm{0} \neq (a,b,c,d) \in \mathbb{C}^4$. We refer to \cref{rem2.2}. Let 
\begin{eqnarray*}
\Sigma(X,Y,Z) &=& Y(1+Z) + \beta X(1+Y)+ d + 2(a+b+c)  \\
&&+a (X -XY) + b(Y -YZ) +c (Z -XZ) + d XYZ \\
&=& [d+2(a+b+c) +(1+b) Y + (c+(1-b)Y)Z] \\
&&+ X[\beta + a  + (\beta - a)Y +(-c +dY) Z   ]\\
&=& U + XV \quad \text{~in short, say.}
\end{eqnarray*} 
Now $1+b$ and $1- b$ can not both be zero. So $U\neq 0$. Because $\beta \neq 0, \, \beta +a$ and $\beta -a$ can not both be zero together. So $\beta + a  + (\beta - a)Y \neq 0$ and thus,   $V \neq 0$. Therefore, $\Sigma(X,Y,Z)$ represents a product vector  in the bipartite cut $\{ \{1\}, \{2,3\}\}$ of  $\mathcal{H}$ if and only if there is a $\gamma \neq 0$ for which $U = \gamma V$ so that $\Sigma(X,Y,Z) = (\gamma + X) V(Y,Z)$. Now $V(Y,Z) = q(Y) r(Z)$ for some polynomials $q$ and $r$ if and only if for some $\delta$ 
\[-c +dY = \delta (\beta +a + (\beta -a) Y).\]
In that case $V(Y,Z)= (\beta +a + (\beta -a) Y) (1+ \delta Z).$ 
If all these conditions hold then we may put 
\[p(X) = \gamma +  X,\quad q(Y) = (\beta+ a) +(\beta- a) Y , \quad r(Z) = 1+ \delta Z.\]
We may write the consistancy conditions for a fixed $\beta$. 
For some $0 \neq \gamma \in \mathbb{C}$ and $\delta \in \mathbb{C}$
\begin{eqnarray*}
d +2(a+b+c) &=& \gamma (\beta +a), \\
c &=& - \gamma c,\\
1-b  &=& \gamma d,\\
-c &=& \delta (\beta +a),\\
d &=& \delta (\beta - a).
\end{eqnarray*}
Consider $a = \beta -2, \, b=1, \, c=0=d$. Then for $\beta\neq1$, we have $\gamma =1$ and $\delta =0$. The product vector is given by 
\[(1+X) (2\beta -2 + 2Y).\]
The corresponding normalised product vector is 
\[ \ket{\zeta_\beta} = \frac{1}{\sqrt{ 2 (|\beta -1|^2 +1) }} (\ket{0} + \ket{1}) \otimes ((\beta -1)\ket{0} + \ket{1})\otimes \ket{0}.\] 
Now $\ket{\zeta_\beta}$s are different for different $\beta$s. So $\mathcal{S}_{\mathcal{V},1,2}$ has infinitely many product vectors i.e. it is not QCES. 
\end{proof}

\subsection{Set up related to Parthasarathy space $\mathcal{S_P}$}
We begin with the bipartite case of equal dimension $d \ge 2$.
\begin{theorem} \label{th4.4}
Let $k=2, \,d_1=d_2=d \ge 2$. The span $\mathcal{S}_{\mathcal{P}, 0, \infty}$ of  $\mathcal{S_P},\, \ket{z_0}$ and $\ket{z_\infty}$ has infinitely many product vectors. 
\end{theorem}
\begin{proof} Let  $ a \in \mathbb{C}$. Then $\ket{\zeta_a}=(\ket{0}+ a \ket{d-1}) \otimes (\ket{0} - a \ket{d-1})\in \mathcal{S}_{\mathcal{P}, 0, \infty}$. Also it is a product vector.

\end{proof}
\begin{rem}
 For $d$ odd, $d =2p+1$ with $p \ge 1$, say,  we have another proof of \cref{th4.4} above. Let $a \in \mathbb{C}$.
\begin{eqnarray*}
\ket{\zeta_a}&=&\left(\ket{0} + a \ket{p} + \frac{a^2}{2} \ket{2p}\right) \otimes \left(\ket{0} - a \ket{p} + \frac{a^2}{2} \ket{2p}\right)\\
 &=& \ket{00} + \frac{a^4}{4} \ket{2p, 2p} +a \left(-\ket{0,p} + \ket{p,0}\right)
 + \frac{a^2}{2} \left( \ket{2p,0} + \ket{0,2p} -2 \ket{p,p}\right)\\
 && + \frac{a^3}{2} \left(\ket{p,2p} - \ket{2p,p} \right)\\
&=& \ket{z_0} + \frac{a^4}{4} \ket{z_\infty} + \ket{\eta},
\end{eqnarray*} 
where $\ket{\eta} \in \mathcal{S_P}$. Hence there are infinitely many product vectors in  $\mathcal{S}_{\mathcal{P}, 0, \infty}$.
 \end{rem}
 \begin{theorem}
For $\mathbb{C}^3 \otimes \mathbb{C}^3$ system, $0 \neq \lambda \in \mathbb{C}$, the space $\mathcal{S}_{\mathcal{P},\lambda,\infty}$  
 has infinitely many product vectors. 
\end{theorem}

\begin{proof}
In view of \cref{rem3.2} (ii), it is enough to consider the case $\lambda=1$.  Let $0<\alpha <1$ and $ q \neq 0$  be real.  Then
\begin{eqnarray*}
\left(1 + X + q X^2 \right) \left(1 + Y + \frac{\alpha}{q} Y^2 \right) &=& 1 + X + Y  + qX^2 + XY  + \frac{\alpha}{q} Y^2 \\
&&+ \frac{\alpha}{q} XY^2   + q X^2 Y + \alpha X^2 Y^2.
\end{eqnarray*}
The corresponding product vectors $\left(\ket{0}+\ket{1}+q \ket{2}\right) \otimes \left(\ket{0}+ \ket{1}+\frac{\alpha}{q}\ket{2}\right)$ are  in  $\mathcal{S}_{\mathcal{P},1,\infty}$ if 
\begin{eqnarray*}
\frac{\alpha}{q} + q &=& 2, 
\end{eqnarray*}
which gives  solutions  $ q = 1 \pm \sqrt{1- \alpha}$ for $ 0 < \alpha <1$. The corresponding neat product vectors  are 
\begin{eqnarray*}
&&\left( \ket{0} + \ket{1} + (1 - \sqrt{1- \alpha}) \ket{2} \right) \otimes \left( \ket{0} + \ket{1} + (1 + \sqrt{1- \alpha})\ket{2} \right),\\
&&\left( \ket{0} + \ket{1} + (1 + \sqrt{1- \alpha}) \ket{2} \right) \otimes \left( \ket{0} + \ket{1} + (1 - \sqrt{1- \alpha}) \ket{2} \right).
\end{eqnarray*}
with $ 0 < \alpha <1$. Hence there are infinitely many product vectors in $\mathcal{S}_{\mathcal{P},1,\infty}$.
\par For general nonzero lambda the technique is similar to the one indicated by \cref{rem3.2}. 
\end{proof}

We now come to other kind of  double perturbation of $\mathcal{S}_P$ in the general multipartite set-up.
\begin{rem}\label{rem4.2}
\par  Let $\lambda, \mu \in \mathbb{C}$ with $\mu \neq \lambda \neq 0$. Then $\mu = \lambda \alpha$ for a unique $ \alpha \in \mathbb{C},\, \alpha\neq 1.$ We will aim at
 product vectors  in these doubly perturbations of $\mathcal{S}_P$ by the corresponding van der Monde vectors.
\begin{enumerate}[(i)]
\item Now 
\[\ket{z_\lambda} = \sum_{n=0}^{N'} \lambda^n\left( \sum_{\bm{i}: |\bm{i}|=n} \ket{\bm{i}} \right),  \]
and  
\[ \ket{z_\mu} = \sum_{n=0}^{N'} \mu^n\left( \sum_{\bm{i}: |\bm{i}|=n} \ket{\bm{i}} \right) = \sum_{n=0}^{N'} \lambda^n \alpha^n\left( \sum_{\bm{i}: |\bm{i}|=n} \ket{\bm{i}} \right)  \]
\item We proceed to figure out if the span $\mathcal{S}_{P, \lambda, \mu}$ of $\ket{z_\lambda}, \ket{z_\mu}$ and $\mathcal{S}_P$ has infinitely many product vectors. 
We first note that for $ \beta \in \mathbb{C} \backslash\{0,-1\}$,
\[ \ket{z_\lambda} + \beta \ket{z_\mu} = (1+ \beta) \ket{\bm{0}} + \lambda^{N'} (1+ \beta\alpha^{N'}) \ket{\bm{d}'}
+ \sum_{n=1}^{N'-1} \lambda^n (1+ \beta\alpha^n) \sum_{\bm{i}:|\bm{i}| =n} \ket{\bm{i}}.\]
Now 
\[\ket{\zeta_\beta} =\frac{1}{1+\beta}\left( \ket{z_\lambda} + \beta \ket{z_\mu} + \sum_{n=1}^{N'-1} \sum_{|\bm{i}|: |\bm{i}|= n} a_{\bm{i}}' \ket{\bm{i}} \right),\] 
where $\sum_{\bm{i}:|\bm{i}|=n} a_{\bm{i}}' =0$ for $1 \le n \le {N'}-1$, is an element of the linear span of $\ket{z_\lambda}, \, \ket{z_\mu}$ and $\mathcal{S}_P$ to within  a scalar multiple, that is not in the linear span of $\ket{z_\mu}$ and  $\mathcal{S}_P$, or of $\ket{z_\lambda}$ and $\mathcal{S}_P$. 
Let $a_{\bm{i}}' = \lambda^{|\bm{i}|}(1+\beta) a_{\bm{i}}, \, {\bm{i}} \in \mathcal{I}, \, 1 \le |{\bm{i}}| \le N'-1$. Then 
\[ \ket{\zeta_\beta} = \ket{\bm{0}} + \lambda^{N'}   \frac{(1+ \beta\alpha^{N'})}{1+\beta}\ket{\bm{d}'}
+ \sum_{n=1}^{N'-1} \lambda^n \frac{(1+ \beta\alpha^n)}{1+\beta} \sum_{\bm{i}:|\bm{i}| =n} \ket{\bm{i}} +  \sum_{n=1}^{N'-1} \lambda^n \sum_{\bm{i}: |\bm{i}|= n} a_{\bm{i}} \ket{\bm{i}}.\]
\item We now use \cref{lem3.1} freely. It is enough to consider the case $\lambda=1$ and we do so.   As in \cref{rem3.2} above,  $ \ket{\zeta_\beta} = \ket{\xi} $ if and only if 
\begin{eqnarray}
1  &=& \prod_{j=1}^k b_{j0} \label{eq:o}\\
\frac{1+\beta\alpha^{N'}}{1+\beta} &=&  \prod_{j=1}^k b_{j,d_j-1} \label{eq:N}\\
\# \mathcal{I}_n \frac{1+ \beta \alpha^n}{1+\beta} &=&\sum_{\bm{i} \in \mathcal{I}_n} \prod_{j=1}^k b_{j i_j}, \quad  1 \le n \le N'-1. \label{eq:n}
\end{eqnarray}
The problem is equivalent to finding 
\[ \bm{b} = [\bm{b}_{j}]_{j=1}^k =[[b_{j t}]]_{\substack{1 \le j \le k \\ 0 \le t \le d_j -1}},\]
that satisfies Equations \eqref{eq:o} -- \eqref{eq:n}.

One needs to check consistency of the equations under different parameters. 
\item  We are not able to prove them in this generality.  
\item We consider the case that Consistency conditions in (iii) are satisfied for $\beta$.  We take $\ket{\zeta_\beta}$ to be in its neat form. (4.3) with $n=1$ gives that 
\[k\frac{1+\beta \alpha}{1+\beta} = \sum_{j=1}^k b_{j,1}. \]
\item Now consider the case when $\beta, \, \beta'$  are satisfying the requirements in (iii) and part (v) above. Because $\alpha\neq 1, \,\dfrac{1+\beta \alpha}{1+\beta}$ is not equal to $\dfrac{1+\beta' \alpha}{1+\beta'}$. So we have that the corresponding $b_{j,1}$'s or $\beta$  have sum not equal to the one for those for $\beta'$. Hence $\ket{\zeta_\beta}\neq \ket{\zeta_{\beta'}}$. 

\item We now proceed  to work out details for $k=2, \,d_1 =d_2=3$ and $k=3,\, d_j=2$ for $j=1,\,2,\, 3$ in the next  few results. 
\end{enumerate}

\end{rem}
\begin{theorem}
For $\mathbb{C}^3 \otimes \mathbb{C}^3$,  for $\lambda, \mu \in \mathbb{C},\, 0 \neq \lambda \neq \mu $, the space $\mathcal{S}_{P,\lambda,\mu}$ contains infinitely many product vectors. 
\end{theorem}
\begin{proof}
 In view of \cref{rem3.2}, it is enough to confine our attention to the case $\lambda=1$.
 For $0 \le s \le 2$, we put $b_s =b_{1,s}$ and $c_s=b_{2,s}$. Then  the equations \eqref{eq:o} -- \eqref{eq:n} take the form 
\begin{eqnarray}
b_0  c_0 &=& 1, \label{eq:oa}\\
b_0  c_1 + b_1  c_0 &=& \frac{2(1 + \beta \alpha)}{1+ \beta}, \label{eq:Na}\\
b_0  c_2 + b_1c_1 + b_2  c_0 &=& \frac{3 (1 + \beta \alpha^2)}{1+ \beta}, \label{eq:ia}\\
b_1  c_2 + b_2  c_1 &=& \frac{2 (1 + \beta \alpha^3)}{1+ \beta}, \label{eq:ib}\\
b_2c_2 &=&  \frac{1 + \beta \alpha^4}{1+ \beta}.  \label{eq:na}
\end{eqnarray}
We confine our attention to the case when $1 + \beta \alpha^s \neq 0$, where $s=0,\,1,\,2,\,3,\,4$.  We may take  $b_0=c_0=1$ satisfying equation \eqref{eq:oa}. To begin with, choose $b_2 \neq 0$ arbitrarily, and set $c_2 = \frac{ 1 + \beta \alpha^4}{b_2(1+\beta)}$ using \eqref{eq:na}. Now choose $b_1$ arbitrary and use \eqref{eq:Na} to choose $c_1 = \frac{2(1 + \beta \alpha)}{1+\beta} -b_1$.  From \eqref{eq:ia} we get, 
\begin{equation}\label{eq:ia'}
\frac{ 1 + \beta \alpha^4}{b_2}+ b_2(1+ \beta) +b_1(2(1 + \beta \alpha) -b_1(1+\beta)) = 3 (1 + \beta \alpha^2).
\end{equation}
From \eqref{eq:ib} we get 
\begin{equation}\label{eq:ib'}
\frac{b_1 (1 + \beta \alpha^4)}{b_2}  + b_2 (2(1 + \beta \alpha) -b_1(1+\beta)) = 2 (1 + \beta \alpha^3)
\end{equation}
Let $b_1 = u b_2$. Then  \eqref{eq:ib'} takes the form 
\begin{equation}\label{eq:ib''}
u((1 + \beta \alpha^4) - b_2^2 (1+\beta)) = 2 (1 + \beta \alpha^3) -2 b_2 (1+\beta \alpha).
\end{equation}
Similarly \eqref{eq:ia'} takes the form 
\begin{equation}\label{eq:ia''}
(1 + \beta \alpha^4) +b_2^2(1+ \beta) +ub_2^2(2(1 + \beta \alpha) - ub_2(1+\beta)) =3 b_2 (1 + \beta \alpha^2).
\end{equation}
{\bf Equivalent aim:} In view of \cref{rem4.2}(vi), the proof will be complete if we show that for an infinite set $S$ of $\beta$'s in $\mathbb{C}$, \eqref{eq:ib''} and \eqref{eq:ia''} have a solution $\nu_\beta =(b_2 ,u)$.  We shall show this in a few steps. 
\begin{enumerate}[(i)]
\item By \eqref{eq:ib''} $u = 0$ forces 
\[b_2 = \frac{1+\beta \alpha^3}{1 + \beta \alpha},\]
and substituting this obvious solution of \eqref{eq:ib''} in \eqref{eq:ia''} we must have 

\[ (1 + \beta \alpha^4) + \left(\frac{1+\beta \alpha^3}{1 + \beta \alpha}\right)^2 (1 + \beta) = 3\left( \frac{1+\beta \alpha^3}{1 + \beta \alpha}\right) (1+ \beta \alpha^2).\]
Simplifying we get  a cubic equation in $\beta$ if $\alpha \neq 0$ and a linear equation (viz., $\beta -1 =0$)  in $\beta$ if $\alpha =0$. So it has at most three solutions. We can ignore them and confine our attention to solutions for which $u \neq 0$.
\item The right hand side of \eqref{eq:ib''} is zero if and only if $b_2 = \frac{1+\beta \alpha^3}{1 + \beta \alpha}$. In that case, left hand side of \eqref{eq:ib''} is 
\begin{eqnarray*}
&=& u \left[ (1 + \beta \alpha^4) - \left(\frac{1+\beta \alpha^3}{1 + \beta \alpha}\right)^2 (1 + \beta)\right]\\
&=& \frac{u \beta (1-\alpha)^3 (1 + \alpha) }{(1 + \beta \alpha)^2} [-1 + \beta \alpha^2], \quad \text{ after simplification.}
\end{eqnarray*}
This is zero ( using $u \neq 0,\, \beta \neq 0 \neq 1- \alpha$) if and only if $\alpha =-1$ or $\beta \alpha^2=1$. We can further combine our attention to those $\beta$'s for which $\beta \alpha^2 \neq 1$. 
\item We now come to the case $\alpha = -1$ and $u \neq 0$ is arbitrary. Then $b_2 =1$. We substitute this in  \eqref{eq:ia''} to obtain 
 \[ 1 + \beta +1 + \beta +u ( 2(1 - \beta) - u (1 + \beta )) = 3( 1 + \beta ).\]
 i.e. 
 \[ u^2 ( 1 + \beta ) - 2u (1 - \beta ) + (1 + \beta ) =0.\]
 We ignore the values $\beta =\pm1$. Then the above equation has solutions 
 \[u = \frac{1-\beta}{1+\beta} \pm \frac{2 \sqrt{-\beta}}{1+\beta}\]
 i.e. 
 \[ b_1 = \frac{1- \beta \pm 2 \sqrt{-\beta}}{1 + \beta}.\]
 Here for a non-zero complex number $z =r \exp(\imath \theta),\, -\pi < \theta \le \pi; \,\sqrt{z}$ is taken to be $r^{\frac{1}{2}} \exp\left(\imath \frac{\theta}{2}\right)$.
 So 
 \[c_1 = \frac{2(1- \beta)  - ( 1 -\beta \pm 2 \sqrt{-\beta})}{1+\beta} =\frac{1 - \beta \mp 2 \sqrt{-\beta}}{1+\beta}.\]
 Moreover, $c_2 = 1 $. Thus we have the neat product vectors 
 \[ \left(\ket{0} + \frac{1 -\beta + 2 \sqrt{-\beta}}{1+\beta} \ket{1} + \ket{2}\right) \otimes \left(\ket{0} + \frac{1 -\beta - 2 \sqrt{-\beta}}{1 + \beta}\ket{1} +  \ket{2}\right), \]
 and 
  \[ \left(\ket{0} + \frac{1 -\beta - 2 \sqrt{-\beta}}{1+\beta} \ket{1} + \ket{2}\right) \otimes \left(\ket{0} + \frac{1 -\beta + 2 \sqrt{-\beta}}{1 + \beta}\ket{1} +  \ket{2}\right). \]
  Hence for $0 \neq \lambda \in \mathbb{C}$ there are infinitely many product vectors in $\mathcal{S}_{P, \lambda, -\lambda}$. 

\item Let, for $x \in \mathbb{C}$,
\begin{eqnarray*}
v(x) &=& x^2 - \frac{1+\beta \alpha^4}{1 + \beta }, \quad \text{ and }\\
r(x) &=& 2 \frac{1+\beta \alpha}{1+\beta} \left(x - \frac{1+\beta \alpha^3}{1 + \beta \alpha}\right). 
\end{eqnarray*}
Then \eqref{eq:ib''} can be rewritten as 
\[uv(x) =r(x), \quad \text{ with } x=b_2,\]
which, under the assumptions set out in (i) and (ii), becomes 
\begin{equation}\label{eq4.13}
x^2 \neq \frac{1+\beta \alpha^4}{1 + \beta }, \quad x \neq \frac{1+\beta \alpha^3}{1 + \beta \alpha}, \quad u = \frac{r(x)}{v(x)}.
\end{equation}
\item We combine \eqref{eq4.13} with \eqref{eq:ia''} to replace  \eqref{eq:ia''}  by 
\begin{eqnarray}\label{eq:ia'''}
p(x) &=&(1+ \beta \alpha^4) + x^2 (1+\beta) + \frac{r(x)}{v(x)} x^2 . 2(1 + \beta \alpha) \nonumber\\
&&- \left( \frac{r(x)}{v(x)} \right)^2 x^3 (1+\beta) -3x (1+\beta \alpha^2) \\
&=& 0.\nonumber
\end{eqnarray}
Now $p\left(\frac{1+\beta \alpha^3}{1 + \beta \alpha}\right) =0$ if and only of 
\[1+\beta\alpha^4 +\left(\frac{1+\beta \alpha^3}{1 + \beta \alpha}\right)^2(1+\beta) - 3\frac{1+\beta \alpha^3}{1 + \beta \alpha}(1+\beta\alpha^2) =0,\]
if and only if 
\[(1+\beta\alpha^4) (1+\beta\alpha)^2 +(1+\beta\alpha^3)^2 (1+\beta) - 3 (1+\beta\alpha^3)(1+\beta\alpha^2)(1+\beta\alpha)=0.\]
We discard the three or less values of $\beta$ given by this polynomial in $\beta$ and proceed with the discussion with $p\left(\frac{1+\beta \alpha^3}{1 + \beta \alpha}\right) \neq0$. 
\item Let $q(x) =\frac{1}{1+\beta} (v(x))^2 p(x)$. Then  
\begin{eqnarray*}
q(x) &=& \left(x^2 + \frac{1+\beta \alpha^4}{1 + \beta}\right) \left(x^2 - \frac{1+\beta \alpha^4}{1 + \beta }\right)^2\\
&& + \left( 2 \frac{1+\beta \alpha}{1 + \beta}\right)^2 \left(x - \frac{1+\beta \alpha^3}{1 + \beta\alpha }\right)\left(x^2 - \frac{1+\beta \alpha^4}{1 + \beta}\right)x^2\\
&& - \left( 2 \frac{1+\beta \alpha}{1 + \beta}\right)^2 \left(x - \frac{1+\beta \alpha^3}{1 + \beta\alpha }\right)x^3 - 3x \left( \frac{1+\beta \alpha^2}{1 + \beta}\right)  \left(x^2 - \frac{1+\beta \alpha^4}{1 + \beta}\right)^2 \\
&=& x^6 +  \left(\frac{1+\beta \alpha^4}{1 + \beta}\right)^3 + x s(x), \quad \text{ say, } 
\end{eqnarray*}
where $s(x)$ is a polynomial of degree $\le 4$. So $q$ can be thought of as a function of $x$ varying in $\mathbb{C}$ (after obvious extension to omitted points as in \eqref{eq4.13} in (iv) above). 
Therefore, $q$ has six zeros, say $\delta_j^{(\beta)}, \, 1 \le j \le 6$ counted with multiplicity and 
\[\prod_{j=1}^6 \delta_j^{(\beta)} = \left(\frac{1+\beta \alpha^4}{1 + \beta}\right)^3.\]
As a consequence, $\delta_j^{(\beta)}  \neq 0,\, 1 \le j \le 6$. Now 
\[ q\left(\frac{1+\beta \alpha^3}{1 + \beta\alpha}\right) = v \left(\frac{1+\beta \alpha^3}{1 + \beta\alpha}\right) p\left(\frac{1+\beta \alpha^3}{1 + \beta\alpha}\right) \neq 0.\]
Also, if at all,  ${\delta_j^{(\beta)}}^2 = \frac{1+\beta \alpha^4}{1 + \beta}$ for some $j$, say $j'$, then 
\[0= q\left( \delta_{j'}^{(\beta)}\right) = -\left(\frac{2(1+\beta \alpha)}{1 + \beta}\right)^2 \left( \delta_{j'}^{(\beta)} - \frac{1+\beta \alpha^3}{1 + \beta\alpha}\right)^2{\delta_{j'}^{(\beta)}}^3.\]
So $\delta_{j'}^{(\beta)} = \frac{1+\beta \alpha^3}{1 + \beta\alpha}$ or $0$. But as we have seen earlier $\delta_{j'}^{(\beta)} \neq 0$. Also $\beta$'s satisfying $\left(\frac{1+\beta \alpha^3}{1 + \beta\alpha}\right)^2 = \frac{1+\beta \alpha^4}{1 + \beta}$ has already been taken out from the consideration. Hence any of the $\delta_{j}^{(\beta)}$'s work as $b_2$ to give us a corresponding $\nu_\beta$. This accomplishes the equivalent aim. 

\item To elaborate the last statement in  (vi) above,  we obtain, for each $j$,  with $1 \le j \le 6$, a   neat product vector, say, $\ket{\zeta_{\beta,j}}$, thus we have the  non-empty set $Z_\beta$ of neat product vectors, viz.,  $\{\ket{\zeta_{\beta,j}}: 1 \le j \le 6\}$ of cardinality $\le6$.
 Let $S$ be the set of surviving  $\beta$'s as per (i) to (vi) above. Then $S$ is infinite. Let $\beta \neq \beta'$ be in $S$. 
Then by \cref{rem4.2}(vi), $Z_\beta \cap Z_{\beta'} = \emptyset$.  Hence the set $\bm{Z} = \bigcup_{\beta\in S} \bm{Z}_\beta $ is infinite.  
\end{enumerate}
This completes the proof.  
\end{proof}
\begin{theorem}
For $\mathbb{C}^2 \otimes \mathbb{C}^2 \otimes \mathbb{C}^2$,  for $\lambda, \mu \in \mathbb{C},\, 0 \neq \lambda \neq \mu $, the space $\mathcal{S}_{P,\lambda,\mu}$ contains infinitely many product vectors. 
\end{theorem}
\begin{proof}
We have to count the number of product states in the span of $\mathcal{S_P}, \ket{z_\lambda}$ and  $\ket{z_\mu}$ using the notation as earlier from \cref{rem4.2} onwards and  in \cref{rem3.2}. We have $N'=3$.   Consider $0 \neq \beta\in\mathbb{C}$ such that $1 + \beta \alpha^s \neq 0, \, s=1,2,3$. 
Furthermore we may take the value $b_{j0} =1,\, j =1,2,3$ to try to get a neat product vector $(\ket{0} + a \ket{1}) \otimes (\ket{0} + b \ket{1}) \otimes (\ket{0} + c \ket{1})$ where $a=b_{11},\, b=b_{21},$ and $c=b_{31}$. The consistency conditions take the following form 
\begin{align*}
3\left(\frac{1+\beta \alpha}{1+\beta}\right) &= a+b +c,\\
3 \left(\frac{1+\beta \alpha^2}{1+\beta} \right)&= ab+bc +ca,\\
\frac{1+\beta \alpha^3}{1+\beta} &= ab c. 
\end{align*}
So $a,\,b,\,c$ are roots of the equation 
\[z^3 - 3\left(\frac{1+\beta \alpha}{1+\beta}\right) z^2 + 3\left(\frac{1+\beta \alpha^2}{1+\beta}\right) z - \left(\frac{1+\beta \alpha^3}{1+\beta}\right) =0.\]
Because $\alpha\neq 1$, for distinct values of permitted $\beta$ the corresponding coefficient-tuples in the above equation are distinct, and therefore, the correspinding sets $\{a,b,c\}$ are distinct. Hence there are infinitely many product vectors in the span of $\ket{z_\lambda}, \, \ket{z_\mu}$, and $\mathcal{S_P}$.

\end{proof}

\section{Conclusion}
\par We have studied the problem of identifying and enumerating  pure product states  in a  perturbation space $T$  of a maximal dimension  completely entangled subspace $\mathcal{S}$ in a  multipartite system $\mathcal{H}$ of dimension  $D$ by a product vector $\ket{\xi}$ in $\mathcal{S}^\perp$. We distinguish them via  triples $(D, t,\tau)$ where $t$ and $\tau$  are the  dimension and product index of $T$ respectively. We find that many related double perturbation spaces have infinitely many pure product states and speculate if all such spaces do have the property.

\section*{Acknowledgement}
RS acknowledges financial support from {\sf DST/ICPS/QuST/Theme-2/2019/}General Project number {\sf Q-90}. The authors thank Tal Mor and  Andreas Winter for useful discussion over email.

\bibliographystyle{acm}
\bibliography{biblio}

\end{document}